\pdfoutput=1
\documentclass[12pt]{article}

\usepackage{lmodern}

\DeclareFontFamily{T1}{calligra}{}
\DeclareFontShape{T1}{calligra}{m}{n}{<->s*[1.44]callig15}{}
\DeclareMathAlphabet\mathcalligra   {T1}{calligra} {m} {n}
\DeclareMathAlphabet\mathzapf       {T1}{pzc} {mb} {it}
\DeclareMathAlphabet\mathchorus     {T1}{qzc} {m} {n}
\DeclareMathAlphabet\mathrsfso      {U}{rsfso}{m}{n}
\DeclareMathAlphabet\mathfrcal      {T1}{frcursive}{m}{it}
\DeclareFontFamily{T1}{frcursive}{}
\DeclareFontShape{T1}{frcursive}{m}{n}{<->s*[1.44]callig15}{}
\DeclareMathAlphabet\mathfrcal      {T1}{frcursive}{m}{it}

\usepackage{amsmath}
\usepackage{amssymb}
\usepackage{amsthm}
\usepackage{graphicx}
\usepackage[dvipsnames]{xcolor}
\usepackage[english]{babel}
\usepackage{csquotes}
\numberwithin{equation}{section}
\usepackage{array}
\usepackage{mathtools}
\usepackage{dsfont}
\usepackage{mathrsfs}
\usepackage{tikz}\usetikzlibrary{matrix,fit}
\usepackage{varwidth}
\usepackage{enumerate}
 \usepackage{appendix}
\usepackage{ytableau}
\usepackage{simpler-wick}
\usepackage{dsfont}
\usepackage{slashbox}

\usepackage{caption}

\usepackage[compat=1.1.0]{tikz-feynman}

\usepackage[margin=1in]{geometry}
\usepackage{nicematrix}

\usepackage[
    backend=bibtex,
    style=alphabetic,
maxbibnames=99,
giveninits=true,
minalphanames=1,
maxalphanames=3,url=false
  ]{biblatex}
\usepackage{mathabx}
\usepackage{empheq}

\usepackage{tabularx}

\usepackage{setspace}

\usepackage{fontawesome5}

\usepackage{slashed}
\usepackage{upgreek}

\usepackage{tabstackengine}

\usepackage{wrapfig}
\usepackage[abs]{overpic}

\usepackage{float}

\usepackage{mathtools}

\DeclarePairedDelimiter\ket{\lvert}{\rangle}

\fixTABwidth{T}

\usepackage{accents}

\newcommand{\CP}{\mathds{CP}}

\newcommand{\bea}{\begin{equation}}
\newcommand{\eea}{\end{equation}}
\newcommand{\bear}{\begin{eqnarray}}
\newcommand{\eear}{\end{eqnarray}}
\newcommand{\bearr}{\begin{eqnarray*}}
\newcommand{\eearr}{\end{eqnarray*}}

\newtheorem{prop}{Proposition}
\newtheorem{lem}{Lemma}

\newdimen\mytextwidth
\newcommand\rem[2][cyan!40!green]{\noindent\nobreak\hfil\penalty1000\hfilneg
\mytextwidth=\linewidth\advance\mytextwidth by 2mm
\begin{tikzpicture}[baseline=-\the\dimexpr\fontdimen22\textfont2\relax]\node[outer sep=0pt,draw=black,fill=#1,fill opacity=1,text opacity=1,rectangle,rounded corners]{\begin{varwidth}{\mytextwidth}\textcolor{white}{#2}\end{varwidth}};
\end{tikzpicture}\allowbreak
}

\newcommand\whiterem[2][white!]{\noindent\nobreak\hfil\penalty1000\hfilneg
\mytextwidth=\linewidth\advance\mytextwidth by 2mm
\begin{tikzpicture}[baseline=-\the\dimexpr\fontdimen22\textfont2\relax]\node[outer sep=0pt,draw=black,fill=#1,fill opacity=1,text opacity=1,rectangle,rounded corners,line width=1.5pt]{\begin{varwidth}{\mytextwidth}\textcolor{black}{#2}\end{varwidth}};
\end{tikzpicture}\allowbreak
}

\renewbibmacro{in:}{}

\usepackage{mdframed}

\newbibmacro{string+doiurlisbn}[1]{%
  \iffieldundef{doi}{%
    \iffieldundef{url}{%
      \iffieldundef{isbn}{%
        \iffieldundef{issn}{%
          #1%
        }{%
          \href{http://books.google.com/books?vid=ISSN\thefield{issn}}{#1}%
        }%
      }{%
        \href{http://books.google.com/books?vid=ISBN\thefield{isbn}}{#1}%
      }%
    }{%
      \href{\thefield{url}}{#1}%
    }%
  }{%
    \href{http://dx.doi.org/\thefield{doi}}{#1}%
  }%
}

\DeclareFieldFormat{title}{\usebibmacro{string+doiurlisbn}{\mkbibemph{#1}}}
\DeclareFieldFormat[article,incollection]{title}%
    {\usebibmacro{string+doiurlisbn}{\mkbibquote{#1}}}

\newmdenv[
  topline=false,
  bottomline=false,
  rightline=false,
  linewidth=2pt,
  skipabove=\topsep,
  skipbelow=\topsep
]{siderules}

\newmdenv[
  topline=false,
  bottomline=false,
  linewidth=2pt,
  skipabove=\topsep,
  skipbelow=\topsep
]{siderulesright}

\makeatletter
\renewcommand{\@seccntformat}[1]{\csname the#1\endcsname.\quad}
\makeatother

\usepackage{setspace}
\usepackage{xpatch}

\makeatletter
\renewcommand{\@chap@pppage}{
  \clear@ppage
  \thispagestyle{plain}
  \if@twocolumn\onecolumn\@tempswatrue\else\@tempswafalse\fi
  \null\vfil
  \markboth{}{}
  {\centering
   \interlinepenalty \@M
   \normalfont
   \MakeUppercase \appendixpagename\par}
  \if@dotoc@pp
    \addappheadtotoc
  \fi
  \vfil\newpage
  \if@twoside
    \if@openright
      \null
      \thispagestyle{empty}
      \newpage
    \fi
  \fi
  \if@tempswa
    \twocolumn
  \fi
}
\makeatother

\definecolor{navycol}{RGB}{100,150,160}
   \definecolor{pinkcol}{RGB}{242,55,55}
   \definecolor{greencol}{RGB}{50,205,50}

   \definecolor{bluecol}{RGB}{30,144,255}

\usepackage{tocloft}

\usepackage{titlesec}

\titleformat*{\section}{\large\bfseries}
\titleformat*{\subsection}{\normalsize\bfseries}
\titleformat*{\subsubsection}{\normalsize\bfseries}
\titleformat*{\paragraph}{\large\bfseries}
\titleformat*{\subparagraph}{\large\bfseries}
\titlespacing{\author}{-5pt}{-5pt}{-5pt}[-5pt]

\makeatletter
\renewcommand\subsubsection{\@startsection{subsubsection}{3}{\z@}
                                     {-3.25ex\@plus -1ex \@minus -.2ex}
                                     {-1.5ex \@plus -.2ex}
                                     {\normalfont\normalsize\bfseries}}
\renewcommand\subsection{\@startsection{subsection}{3}{\z@}
                                     {-3.25ex\@plus -1ex \@minus -.2ex}
                                     {-1.5ex \@plus -.2ex}
                                     {\normalfont\normalsize\bfseries}}                                     
\makeatother

\usepackage{soul}

\usepackage{scalerel}[2016/12/29]

\DeclareFontFamily{U}{solomos}{}
\DeclareErrorFont{U}{solomos}{m}{n}{10}
\DeclareFontShape{U}{solomos}{m}{n}{
  <-> s*[1.1]  gsolomos8r
}{}

   \usepackage{stmaryrd}

\usepackage{tikz}
\usetikzlibrary{arrows.meta}

\usepackage{indentfirst}

\let \savenumberline \numberline
\def \numberline#1{\savenumberline{#1.}}

\usepackage{xurl}
\usepackage{scrextend}

\usepackage{upgreek}

\usepackage{etoolbox}
\patchcmd{\tableofcontents}{\@starttoc}{\vspace{-0.3cm}\@starttoc}{}{}

\usepackage{nicematrix}

\usepackage{array}

\newcounter{Chapcounter}

\newcommand{\chapter}[1] 
{ {\centering          
  \addtocounter{Chapcounter}{1} \Large \underline{\sffamily \texorpdfstring{\textbf{  Chapter \theChapcounter: ~#1}}{Lg}} }   
  \addcontentsline{toc}{section}{ \color{blue} \texorpdfstring{Chapter ~}{Lg}\theChapcounter.\texorpdfstring{~~}{Lg} #1 }    
}

\newcommand{\appendixbig}[1] 
{ {\centering          
   \Large \underline{\sffamily \textbf{  Appendices}} }   
  \addcontentsline{toc}{section}{ \color{blue} Appendices}    
}

\usepackage[hyperfootnotes=]{hyperref}
\hypersetup{
    colorlinks=true,
    linkcolor=blue,
    filecolor=magenta,      
    urlcolor=cyan,
    breaklinks=true
    }

\title{\textbf{Sigma models from Gaudin spin chains} \vspace{0.7cm}}

\author{Dmitri Bykov$^{\,a,\,b,\,c,\,d,\,e}$\footnote{Emails:
 bykov@mi-ras.ru, dmitri.v.bykov@gmail.com} \qquad\qquad Andrew Kuzovchikov$^{\,a,\,b,\,c,\,d}$\footnote{Email:
 andrkuzovchikov@mail.ru}
\\  \vspace{0cm}  \\
{\small $a)$ 
\emph{Steklov
Mathematical Institute of Russian Academy of Sciences,}} \\{\small \emph{Gubkina str. 8, 119991 Moscow, Russia} }\\
{\small $b)$ 
\emph{Institute for Theoretical and Mathematical Physics,}} \\{\small \emph{Lomonosov Moscow State University, 119991 Moscow, Russia}}\\
{\small $c)$ \emph{HSE University, 6 Usacheva str., Moscow 119048, Russia}}\\
{\small $d)$ \emph{Moscow Center of Fundamental and Applied Mathematics,}} \\{\small \emph{ Lomonosov Moscow State University, 119991 Moscow, Russia}}\\
{\small $e)$ \emph{Beijing Institute of Mathematical Sciences and Applications (BIMSA),}} \\{\small \emph{Huairou District, Beijing
101408, China}}
}

\date{}

\begin{document}

\maketitle

\ytableausetup{centertableaux}

\begin{abstract}
We solve the classical and quantum problems for the 1D sigma model with target space the flag manifold $\mathrm{U}(3)\over \mathrm{U}(1)^3$, equipped with the most general invariant metric. In particular, we explicitly describe all geodesics in terms of elliptic functions and demonstrate that the spectrum of the Laplace-Beltrami operator may be found by solving polynomial (Bethe) equations. The main technical tool that we use is a mapping between the sigma model and a  Gaudin model, which is also shown to hold in the $\mathrm{U}(n)$ case. 
\end{abstract}

\newpage
\section{Introduction}

Dualities between different systems play an important role in modern mathematical physics. With their help, it is possible to apply methods from one area in another one. In the present work we will  investigate  the duality between spin chains and one-dimensional sigma models on (co)adjoint orbits of $\mathrm{SU}(N)$ -- the flag manifolds \cite{Bykov_2024,OscCalculusCoadj,IsotropicBykov}. A precursor of a duality of this type arose in the physics-oriented work on the continuum limit of $\mathrm{SU}(2)$ spin chains that leads to the two-dimensional sigma model on $\mathrm{S}^2$~\cite{HaldaneNonlin}. Later, this was generalized to $\mathrm{SU}(N)$ spin chains and two-dimensional sigma models on flag manifolds in~\cite{BykLagEmb,Bykov_2013, Affleck1, Affleck2}. However, the results that we will report in this paper refer to the \emph{one-dimensional} case, where the mapping becomes fully rigorous.

Let us take a closer look at the one-dimensional case. At the classical level the `sigma model' is tantamount to the study of the geodesic flow on the relevant manifold, while quantization leads to the spectral problem for a Laplace-type operator. In the present article, we will use the duality between sigma models on flag manifolds and spin chains to solve both problems. Our main example is $\mathcal{F}(3)= \frac{\mathrm{U}(3)}{\mathrm{U}(1)^3}$ -- the Wallach flag manifold of~$\mathrm{SU}(3)$~\cite{WallachFlag}. 

Geodesic flow on homogeneous spaces is of special interest in modern differential geometry~\cite{Bolsinov}. Solution to geodesic flow equations was explicitly found for a special class of Riemannian metrics on homogeneous spaces of compact Lie groups~\cite{Souris_2023}. As regards the flag manifolds, one should mention that geodesic flow on $\mathcal{F}(3)$ has been proven to be integrable for all invariant metrics~\cite{Paternain}, using a variation of the method proposed in~\cite{Thimm}. However, explicit expressions for the geodesics were not provided. We will fill in this gap in the first part of this paper using the relation to spin chains as well as direct methods\footnote{In this sense, we follow the idea of~\cite{Perelomov_2005}.}.

Much less is known about the spectrum of the Laplace-Beltrami operator. The list of manifolds for which the spectrum has been calculated includes spheres, projective spaces \cite{berger1971spectre, IKEDA} (including the magnetic case in~\cite{Kuwabara}), flat tori \cite{berger1971spectre}, Riemannian symmetric spaces \cite{Wolf1976} and naturally reductive homogeneous spaces~\cite{agricola2025}. A separate body of work concerns the study of first eigenvalues of the Laplace-Beltrami operator on various manifolds~ \cite{Lauret_2018,Bettiol_2022,LingLu,URAKAWA_1979}. For generalized flag manifolds in particular, the first eigenvalues were studied in \cite{Doi_1987}. The spectrum has been calculated for the normal metric on flag manifolds in \cite{Yamaguchi} and for submersion metrics associated with forgetful bundles in \cite{Bykov_2024}. In light of the integrability of geodesic flow equations on $\mathcal{F}(3)$ and the relation of this classical problem to the quantum spectral problem\footnote{See \cite{Camporesi} for details.}, one should expect that the spectrum of the Laplace-Beltrami operator can be calculated for an arbitrary invariant metric on~ $\mathcal{F}(3)$. Indeed, the respective calculation is the second main result of this paper. 

From a more practical standpoint we should mention that the flag manifold $\mathcal{F}(3)$ is relevant for $AdS_4$ compactifications of Type IIA supergravity, cf.~\cite{Tomasiello, KoerberLustTsimpis}. The same flag manifold also features in flat space compactifications of superstring models~\cite{Lust, CastellaniLust} (see also~\cite{Klaput, Matti, Zoupanos} and references therein for more recent work in this direction). In this context the spectrum of the Laplacian is particularly relevant for estimating the masses of the fields corresponding to Kaluza-Klein harmonics.  Another interesting observation is that the class of metrics that we study  is closed w.r.t. Ricci flow, as discussed in~\cite{Bakas}. In particular, Ricci flow for flag manifolds has been described rather explicitly  in~\cite{Grama2009, Grama2020}.

The paper is organized as follows. We start in Section~\ref{flagsec} by reviewing the Lagrangian embedding of an arbitrary flag manifold into a product of Grassmannians, on which the results of the paper are ultimately based. In Section~\ref{SUNSpinChain} we explain the relation between a Gaudin-type model and a 1D sigma model with flag manifold target space. We then proceed in Section~\ref{geodSection} to explicitly solve the geodesic equations on the flag manifold $\mathcal{F}(3)$ in terms of elliptic functions. In the subsequent Section~\ref{F3specsec} we solve the corresponding quantum problem, i.e. we determine the spectrum of the Laplace operator on $\mathcal{F}(3)$ using Bether ansatz for the Gaudin model. In particular, we find that the spectrum is essentially encoded in the Heun polynomials. Finally, in Section~\ref{spectrRecon} we develop an alternative method (called `spectral reconstruction') that allows calculating some eigenvalues without the use of Bethe ansatz. These independent results are used to compare with the general findings of Section~\ref{F3specsec}.

\section{Flag manifolds}\label{flagsec}
In this section, we will review the basic definitions and geometric constructions, concerning flag manifolds. A flag is an ordered set of mutually orthogonal linear subspaces in $\mathbb{C}^N$ 
\begin{align}\label{Lspaces}
    \{\mathsf{L}_i\}_{i=1}^k\;,\quad\text{such  that}\;\, \mathbb{C}^N = \bigoplus_{i=1}^k \mathsf{L}_i\
\end{align}
with prescribed dimensions $n_i = \mathrm{dim}\; \mathsf{L}_i$. The set of all flags for fixed $n_i$'s is known as the flag manifold. It is a homogeneous space\footnote{We refer to \cite{Affleck_2022} and \cite{ArvGeometryOfFlags} for details on the geometry of flag manifolds.} of both  $\mathrm{GL}(N,\mathbb{C})$ and $\mathrm{SU}(N)$. More precisely, it is a (co)adjoint orbit of $\mathrm{SU}(N)$ and, as a quotient space, it has the form
\begin{align}
    \mathcal{F}_{n_1,n_2,\dots,n_k} = \frac{\mathrm{SU}(N)}{\mathrm{S}\left(\mathrm{U}(n_1)\times \mathrm{U}(n_2)\times \dots \times \mathrm{U}(n_k)\right)}\,,
\end{align}
where $n_1+n_2+\dots+n_k=N$ and $\mathrm{S}\left(\dots\right)$ stands for a subgroup of matrices with determinant $1$. We define the full flag manifold as
\bea
\mathcal{F}(N) := \mathcal{F}_{1,1,\dots,1}
\eea 
All other flag manifolds are called partial. Some remarkable examples of partial flag manifolds are the Grassmannians $\mathrm{Gr}(n,N):=\mathcal{F}_{n,N-n}$ and projective spaces $\mathbb{CP}^{N-1}:=\mathcal{F}_{1,N-1}$. They can be described only via the single linear subspace $\mathsf{L}_1$ in~(\ref{Lspaces}) because the second subspace $\mathsf{L}_2 = \mathsf{L}_1^{\perp}$, where $\perp$ stands for an orthogonal complement in $\mathbb{C}^N$. Thus, to describe a point in $\mathrm{Gr}(n,N)$, one can use an $N\times n$ matrix $Z$ of maximal rank. The columns of $Z$ form a basis in $\mathsf{L}_1$. Naturally, this matrix is defined up to an $\mathrm{GL}(n,\mathbb{C})$ transformation. Note that if one requires the columns of $Z$ to be orthogonal and have unit norm, then the redundancy group reduces to $\mathrm{U}(n)$.

As a (co)adjoint orbit, every flag manifold is endowed with the natural symplectic form -- the Kirillov-Kostant-Souriau form \cite{Kirillov_2004}. In the case of projective spaces this is the well-known Fubini-Study form.

From the definition of the flag manifold, one finds that there is a natural embedding 
\begin{align}
    \mathcal{F}_{n_1,n_2\dots,n_k}\hookrightarrow \mathrm{Gr}(n_1,N)\times\mathrm{Gr}(n_2,N)\times\dots\times\mathrm{Gr}(n_k,N)\,,\label{embedFlag}
\end{align}
where one maps a flag $\{\mathsf{L}_i\}_{i=1}^k$ in $\mathcal{F}_{n_1,n_2\dots,n_k}$ into a set $\{\mathsf{L}_i, \mathsf{L}_i^{\perp}\}_{i=1}^k$. Every pair $\{\mathsf{L}_i, \mathsf{L}_i^{\perp}\}$ is thought of as a point in $\mathrm{Gr}(n_i,N)$. It is remarkable that the embedding is \emph{Lagrangian} when the symplectic form on the RHS of (\ref{embedFlag}) is chosen to be a sum of (generalized) Fubini-Study forms on each Grassmannian. So, we arrive at the following statement:
 
\begin{prop}[\cite{BykLagEmb}]\label{lagrtheorem}
       $\mathcal{F}_{n_1,\dots,n_k}$ is a Lagrangian submanifold of $(\mathrm{Gr}(n_1,N)\times\dots\times\mathrm{Gr}(n_k,N), \Omega)$, where the symplectic form $\Omega$ is a sum of Fubini-Study forms and the submanifold is distinguished by the orthogonality condition on the planes corresponding to points of the Grassmannians.
\end{prop}

In fact, there is an even stronger statement:

\begin{prop}[\cite{Bykov_2024}]\label{subbundle}
       Let $\mathcal{X}\subset \mathrm{Gr}(n_1,N)\times\dots\times\mathrm{Gr}(n_k,N)$ be the set of all $k$-tuples of linear subspaces in~$\mathbb{C}^N$ such that their direct sum is $\mathbb{C}^N$. Then $\mathcal{X}$ is symplectomorphic to an open subset of $\mathrm{T}^\ast \mathcal{F}_{n_1,\dots,n_k}$.
\end{prop}
The symplectomorphism can be described as follows. Every Grassmanian $\mathrm{Gr}(n_i,N)$ is parameterized by a matrix $Z_i$ of size $N \times n_i$. One assembles these matrices in a single $N\times N$ matrix
\bea
\mathsf{Z} = \begin{pmatrix}
        Z_1 & Z_2 & \dots & Z_k
    \end{pmatrix}\,.
    \eea
Then, by virtue of the polar decomposition theorem, for a non-degenerate matrix  $\mathsf{Z}$\footnote{The condition $\mathrm{det}\left(\mathsf{Z}\right)\neq 0$ distinguishes $\mathcal{X}$ in Proposition \ref{subbundle}.} there is a decomposition $\mathsf{Z} =\mathsf{U}\mathsf{H}$, where $\mathsf{U}:=\mathsf{Z}\left(\mathsf{Z}^{\dagger}\mathsf{Z}\right)^{-1/2}$ is unitary and $\mathsf{H}:=\left(\mathsf{Z}^{\dagger}\mathsf{Z}\right)^{1/2}$ a positive-definite Hermitian matrix. The matrix $\mathsf{U}$ defines a point in $\mathcal{F}_{n_1,\dots,n_k}$ whereas~$\mathsf{H}$ specifies a point in the fiber of $\mathrm{T}^{\ast}\mathcal{F}_{n_1,\dots,n_k}$.

In fact, the embedding (\ref{embedFlag}) maybe used to describe $\mathcal{F}_{n_1,n_2,\dots, n_k}$. A point in $\mathcal{F}_{n_1,n_2,\dots, n_k}$ is associated with the unitary matrix $\mathsf{U} = \begin{pmatrix}
        \mathrm{U}_1 & \mathrm{U}_2 & \dots & \mathrm{U}_k
    \end{pmatrix}$, where $\mathrm{U}_i$ is a matrix corresponding to $\mathrm{Gr}(n_i, N)$ in the RHS of the embedding. In these terms, the  most general invariant metric on the flag manifold has the form \cite{Arvanitoyeorgos_1993}
\begin{align}
    \mathrm{d}s^2 = \sum_{i<j}\frac{1}{\alpha_{ij}}\mathrm{Tr}\left(\mathrm{U}^{\dagger}_i\mathrm{d}\mathrm{U}_j\mathrm{U}^{\dagger}_j\mathrm{d}\mathrm{U}_i\right)\,,
\end{align}
where $\alpha_{ij} \in \mathbb{R}_{>0}$
In the case of $\mathcal{F}(N),$ we have 
\begin{align}\label{FNmetric}
    \mathrm{d}s^2 = \sum_{i<j}\frac{1}{\alpha_{ij}}|\Bar{u}_i \cdot \mathrm{d}u_j|^2\,,
\end{align}
where $\{u_i \in \mathbb{CP}^{N-1}\}_{i=1}^N$ are unit-normalized and orthogonal, i.e. $\Bar{u}_i \cdot u_j = \delta_{ij}$. We use the notation $\Bar{x}\cdot y := \sum_{k=1}^N \Bar{x}_k y_k$ for the usual scalar product. 

\section{$\mathrm{SU}(N)$ spin chains and full flag manifolds}\label{SUNSpinChain}
In this section, we will discuss the construction of a special $\mathrm{SU}(N)$ spin chain and its relationship to one-dimensional sigma models on full flag manifolds.

The $\mathrm{SU}(N)$ spin chain we are dealing with consists of $N$ sites. For each site, we assign a specific irreducible representation. We refer to this as $\mathrm{Sym}(p)$, which is the $p$-th symmetric power of the fundamental representation of $\mathrm{SU}(N)$ and corresponds to the following Young diagram:
    \begin{gather}\label{SYMdef}
        \mathrm{Sym}(p) =
        \underbrace{
        \begin{ytableau}
               ~ & ~ & \dots & ~
        \end{ytableau}
        }_{p}.
    \end{gather}
Therefore, the full Hilbert space of the model, $\mathrm{Hilb}(p)$, is represented as 
\begin{align}\label{Hilbp}
    \mathrm{Hilb}(p) := \underbrace{
        \begin{ytableau}
               ~ & ~ & \dots & ~
        \end{ytableau}
        }_{p}^1 \otimes \underbrace{
        \begin{ytableau}
               ~ & ~ & \dots & ~
        \end{ytableau}
        }_{p}^2 \otimes\dots \otimes \underbrace{
        \begin{ytableau}
               ~ & ~ & \dots & ~
        \end{ytableau}
        }_{p}^N = \mathrm{Sym}(p)^{\otimes N}\,.
\end{align}
Now we define the Hamiltonian of the system to be\footnote{We always assume summation w.r.t. repeated indices.}
\begin{align}
    \mathcal{H} = \sum_{i<j}\alpha_{ij} \mathrm{S}_{i}^a \mathrm{S}_{j}^a+ \mathrm{const}\,,\label{spinHamN}
\end{align}
where $\mathrm{S}_i^a, \,i=1,2,\dots,N$ are the  generators $\tau_a$ of $\mathfrak{su}(N)$  ($a = 1,2,\dots, \mathrm{dim}\;\mathfrak{su}(N)$) in $\mathrm{Sym}(p)$. We use the normalization $\mathrm{Tr}(\tau^a \tau^b)=\delta^{ab}$  for the generators. Thus, up to a constant, $\mathrm{S}_i^a \mathrm{S}_j^a$ is a quadratic Casimir operator of $\mathfrak{su}(N)$, acting in the tensor product of representations with labels $i$ and $j$. The number `$\mathrm{const}$' is chosen so that the eigenvalue of $\mathcal{H}$, corresponding to the trivial representation in $\mathrm{Hilb}(p)$, is zero.

Remarkably, one can find a classical dynamical system whose quantum version is  the above spin chain~\cite{Affleck_2022} (the simplest way to do this is by using the Schwinger-Wigner quantization technique). The respective action~is
\begin{align}
    \mathcal{S} = \int dt\,\left( i \sum_{j=1}^N\,\Bar{z}_j \cdot \dot{z}_j  - \sum_{k<j}\alpha_{kj} |\Bar{z}_k \cdot z_j|^2 \right)\,, \label{spinLagrSU(N)}
\end{align}
where $z_i \in \mathbb{CP}^{N-1}$, $\dot{f}:=\frac{\mathrm{d}\;}{\mathrm{d}t}f$ and we have also imposed the normalization constraints $\Bar{z}_i \cdot z_i = p$. Using the polar decomposition theorem for $\mathsf{Z} = \begin{pmatrix}
        z_1 & z_2 & \dots & z_N
    \end{pmatrix} = \mathsf{U}\mathsf{H}$ (here $\mathsf{U}$ is unitary and $\mathsf{H}$ Hermitian positive-definite), upon eliminating  $\mathsf{H}$ one recovers\footnote{To be more precise, this is true only in the limit $p\to \infty$. We refer to \cite{Bykov_2024} and \cite{IsotropicBykov} for further details.} the action of a one-dimensional sigma model on $\mathcal{F}(N)$,  given by
\begin{align}
    \mathcal{S}_{\sigma}=\int dt\,\sum_{k<j}\frac{1}{\alpha_{kj}}|\Bar{u}_k \cdot \dot{u}_j|^2\,,
\end{align}    
where $u_i \in \mathbb{CP}^{N-1}$ are the columns of $\mathsf{U}$, so that  $\Bar{u}_k\cdot u_j = \delta_{kj}$.

At the classical level, the above discussion implies that solutions to the e.o.m following from (\ref{spinLagrSU(N)}) are in one-to-one correspondence with geodesics on $\mathcal{F}(N)$ having restricted momenta, where  the restriction is determined by $p$ and disappears in the limit~$p \to \infty$. At the quantum level, we arrive at the following proposition:
\begin{prop}[\cite{Bykov_2024}]\label{spinFlagConnection}
        Let $\mathcal{H}$ be the Hamiltonian~(\ref{spinHamN}) acting in the space~$\mathrm{Hilb}(p)$. Then
    \bea \label{partitionFuncF3}
        \underset{p\to\infty}{\mathrm{lim}}\,\mathrm{Tr}_{\mathrm{Hilb}(p)}(g\,e^{- \tau \mathcal{H}})=\mathrm{Tr}_{\mathrm{L}^2(\mathcal{F}(3))}(g\,e^{- \tau\mathcal{H}_{\mathrm{particle}}})\,, 
    \eea
    where $\tau>0$ is a parameter, $g\in\mathrm{SU}(N)$ is a fixed group element taken in the relevant representation, $\mathcal{H}_{\mathrm{particle}} = -\triangle$, where $\triangle$ is the Laplace-Beltrami operator for the metric~(\ref{FNmetric}) on $\mathcal{F}(N)$. 
\end{prop} 

As an equality of partition functions, (\ref{partitionFuncF3}) implies that 
\bea
\underset{p\to\infty}{\mathrm{lim}}\, \mathrm{Hilb}(p) = \mathrm{L}^2(\mathcal{F}(N))
\eea
and the spectrum of $\mathcal{H}$ coincides with that of $\mathcal{H}_{\mathrm{particle}}$ in the same limit. Proposition~\ref{spinFlagConnection} provides a way to compute the Laplace-Beltrami operator spectrum for $\mathcal{F}(N)$.

Let us now discuss a special class of $\mathrm{SU}(N)$ spin chains, namely the ones arising in the Gaudin model (see \cite{Feigin1994-ag,Mukhin2007-rl} for the general introduction). Here one considers a spin chain consisting of $n$ sites and, for each site $i$, one assigns  a representation $\mathrm{V}_i$ of $\mathrm{SU}(N)$. The Hilbert space of the model is $\bigotimes_{i=1}^n \mathrm{V}_i$. One constructs the following set of Hamiltonians, which are customarily called the Gaudin Hamiltonians:
\begin{align}
    \mathrm{H}_i = \sum_{j \neq i}\frac{\mathrm{S}_{i}^a \mathrm{S}_{j}^a}{z_i - z_j}\,,\quad \text{for}\;\,i=1,\dots,n\,, \label{gaudinHam}
\end{align}
where $z_i$'s $\in \mathbb{R}$ are some numbers, which we refer to as the Gaudin parameters, chosen so that $z_i \neq z_j$ if $i\neq j$, and $\mathrm{S}_i^a$ is a generator $\tau^a$ of $\mathfrak{su}(N)$ in representation $\mathrm{V}_i$. We use the same normalization for $\tau^a$'s as in (\ref{spinHamN}). Just like before, $\mathrm{S}_{i}^a \mathrm{S}_{j}^a$ is a quadratic Casimir operator of $\mathfrak{su}(N)$, acting in $\mathrm{V}_i \otimes \mathrm{V}_j$. One easily checks that $[\mathrm{H}_i, \mathrm{H}_j] = 0$ and $\sum_{i=1}^n \mathrm{H}_i = 0$. We will now consider 
\bea\label{tham}
\mathcal{H}=\sum_{i=1}^n t_i \mathrm{H}_i
\eea
as the Hamiltonian, and it turns out that in this case one has a completely integrable system  called the Gaudin model (see, for example, \cite{gaudin_1976,Jurco1990-jb,GaudinBook}). The spectrum of such spin chain can be computed using the Bethe ansatz, which we will briefly recall below.

It is easy to see that, in the special case $n=3$, all spin Hamiltonians, $\sum_{i < j} \alpha_{ij} \mathrm{S}_i^a \mathrm{S}_j^a$, are combinations of Gaudin Hamiltonians for a suitable choice\footnote{In fact, there is even some freedom in choosing $t_i$'s and $z_i$'s.} of $z_i$'s and $t_i$'s. Thus, the spin chain corresponding to $\mathcal{F}(3)$ is integrable for all possible $\alpha_{ij}$'s at the quantum level. Based on this, one should also expect integrability and the existence of Gaudin-type integrals of motion at the classical level. Below, using the relation to spin chains that we have reviewed, we will find solutions to both the classical and quantum problems for the flag manifold $\mathcal{F}(3)$, equipped with an arbitrary invariant metric.

Before proceeding, let us examine the classical aspects of the considered spin chain~(\ref{spinLagrSU(N)}). The e.o.m. following from the action are\footnote{Here we have used the $\mathrm{U}(1)^{N}$ gauge redundancy to eliminate the Lagrange multipliers corresponding to the constraints $\Bar{z}_i \cdot z_i = p$.} 
\begin{align}
    &i\, \frac{\mathrm{d}\,\;}{\mathrm{d}t}\mathsf{Z}= \mathsf{Z}\times\mathfrak{I}(\mathcal{L})\,,
    \label{ZeqNewForm}\\
    &\text{where}\;\; \mathsf{Z}:= \begin{pmatrix}
        z_1 & z_2 &\dots & z_N
    \end{pmatrix}\,,\nonumber\\
    &\left[\mathcal{L}\right]_{ij} = \begin{cases}
        \Bar{z}_i \cdot z_j\,,\quad i\neq j\\
        \quad \quad 0,\quad i = j
    \end{cases}\nonumber\\
    &\text{and}\quad\left[\mathfrak{I}(\mathcal{L})\right]_{ij} = \alpha_{ij}\left[\mathcal{L}\right]_{ij}\,,\nonumber
\end{align}
we have assumed $\alpha_{ij}=\alpha_{ji} $ and $ \alpha_{ii}=0$. As we will show below, if one finds $\mathcal{L}$ then the solution to (\ref{ZeqNewForm}) can be easily constructed as well. Therefore let us write out the equation on $\mathcal{L}$ that follows from~(\ref{ZeqNewForm}):
\begin{align}
    &i\frac{\mathrm{d}\;}{\mathrm{d}t}\mathcal{L} = \left[\mathcal{L}, \mathfrak{I}(\mathcal{L})\right]\,.\label{eulerEq}
\end{align}
This is the famous Euler equation on $\mathrm{SU}(N)/\mathrm{S}(\mathrm{U}(1)^N)$ (the quotient corresponds to the fact that the diagonal values of $\mathcal{L}$ are zero). Notice that, when $\mathcal{L}$ is pure imaginary, the above equations reduce to Euler's rigid body equations for $\mathfrak{so}(N)$. Integrability of such equations was studied in~\cite{Manakov_1977} and, more generally, in~\cite{MiFom}. In particular, it was proven that the system is completely integrable in the case when the parameters are chosen as 
\begin{equation}
    \alpha_{ij}=\frac{t_i-t_j}{z_i-z_j}\,.
\end{equation}
Clearly, this choice of $\alpha_{ij}$'s corresponds to the case of the Gaudin model via the map described in Section~\ref{SUNSpinChain}. 

Now, suppose that $\mathcal{L}$ is known.
Thus, we may first write the solution to (\ref{ZeqNewForm}) formally as
\begin{align}\label{Zformalsol}
    \mathsf{Z}(t) = \mathsf{Z}(0)\times \mathsf{V}(t)\,,\quad\quad \textrm{where}\quad\quad \mathsf{V}(t)=\mathcal{T}\left[\displaystyle e^{-i\int_{0}^t \mathrm{d}t \,\mathfrak{I}(\mathcal{L})}\right]\,.
\end{align}
However, as we shall now show, $\mathsf{V}(t)$ can be calculated purely algebraically from the data that we already have: 
\begin{lem}
 $\mathsf{V}(t)$ can be expressed explicitly through the matrix of eigenvectors of $\mathcal{L}(t)$.   
\end{lem}

\begin{proof}
    First, notice that the Euler equation~(\ref{eulerEq}) implies that the eigenvalues of $\mathcal{L}(t)$ are constant in time, so that we may write $\mathcal{L}(t)=\mathsf{E}(t) \Lambda \mathsf{E}(t)^\dagger$, with~$\Lambda$ a constant diagonal matrix of eigenvalues and $\mathsf{E}(t)$ the matrix of eigenvectors of $\mathcal{L}(t)$, defined up to right multiplication by a diagonal matrix of phases. We will additionally assume that the spectrum of $\mathcal{L}(t)$ is simple. In this case it is easy to show that the Euler equation~(\ref{eulerEq}) leads to the following equation for $\mathsf{E}(t)$:
\bea\label{Eeq}
-i\,\frac{\mathrm{d}\,\;}{\mathrm{d}t}\mathsf{E}(t)=\mathfrak{I}(\mathcal{L})\,\mathsf{E}(t)+\mathsf{E}(t)\,\mathsf{Q}\,,
\eea
where $\mathsf{Q}$ is diagonal. It can be viewed as a $\mathfrak{u}(1)^N$ gauge connection, where the gauge freedom is related to the fact that the phases of the eigenvectors featuring in $\mathsf{E}(t)$ may be chosen arbitrarily. In fact, since we know the explicit expression for $\mathcal{L}(t)$, we may pick any matrix $\mathsf{E}(t)$ of its eigenvectors and then read off $\mathsf{Q}$ from the above equation: $\mathsf{Q}(t)=-i \mathsf{E}^\dagger \dot{\mathsf{E}}-\mathsf{E}^\dagger\mathfrak{I}(\mathcal{L}) \mathsf{E}$. On the other hand, the solution to~(\ref{Eeq}) is 
\bea
\mathsf{E}(t)=\mathsf{V}(t)^\dagger \mathsf{E}(0)\,\mathrm{exp}\left(i\,\int_0^t \mathsf{Q}(t') dt'\right)
\eea
This allows expressing  $\mathsf{V}(t)$ through $\mathsf{E}(t)$ as follows:
\bea\label{Vsol}
\mathsf{V}(t)=\mathsf{E}(0)\,\mathrm{exp}\left(i\,\int_0^t \mathsf{Q}(t') dt'\right)\,\mathsf{E}(t)^\dagger\,,
\eea
which completes the proof.
\end{proof}

Substituting~(\ref{Vsol}) in~(\ref{Zformalsol}) gives an explicit solution for $\mathsf{Z}(t)$, hence solving the classical spin chain problem. 
Using the polar decomposition theorem, we map $\mathsf{Z}$ to a geodesic in $\mathcal{F}(N)$, which is represented by the unitary matrix $\mathsf{U}= \mathsf{Z}\left(\mathsf{Z}^{\dagger}\mathsf{Z}\right)^{-1/2}$. Clearly, we arrive at
\begin{align}
    \mathsf{U}(t)=\mathsf{U} (0)\times \mathsf{V}(t)\,,\label{geodUsol}
\end{align}
so that the problem of finding geodesics can be solved as soon as one has a solution to the Euler equation (\ref{eulerEq}). Summarizing all of the above, we arrive at\footnote{Strictly speaking, the paper~\cite{MiFom} dealt with geodesic flow on a Lie group (say, $\mathrm{SU}(n)$), but their results easily extend to flag manifolds as well.} 
\begin{prop}[\cite{MiFom}]
    Geodesic flow on $\mathcal{F}(N)$ is completely integrable in the case of invariant metrics that correspond to the Gaudin model, i.e. such that~${\alpha_{ij} = \frac{t_i - t_j}{z_i - z_j}}$.
\end{prop}

\section{Geodesics on $\mathcal{F}(3)$}\label{geodSection}

In this section, we will use the correspondence with the $\mathrm{SU}(3)$ spin chain to describe geodesics on $\mathcal{F}(3)$, equipped with the most general invariant metric
\begin{align}
    \label{F3metric}
    \mathrm{d}s^2 = \frac{1}{\alpha_{12}}|\Bar{u}_1 \cdot \mathrm{d}u_2|^2+\frac{1}{\alpha_{13}}|\Bar{u}_1 \cdot \mathrm{d}u_3|^2+\frac{1}{\alpha_{23}}|\Bar{u}_2 \cdot \mathrm{d}u_3|^2\,,
\end{align}
where $\{u_i \in \mathbb{CP}^2\}_{i=1}^3$ and $\Bar{u}_i \cdot u_j = \delta_{ij}$.

Let us analyze the relevant Euler equation (\ref{eulerEq}). In our case it reads as follows:
\begin{align}
    &i\,\dot{a} = \left(\alpha_{23} - \alpha_{13}\right)b\Bar{c}\,,\nonumber\\
    &i\,\dot{b} = \left(\alpha_{23} - \alpha_{12}\right)ac\,,\label{equatiForCIn}\\
    &i\,\dot{c} = \left(\alpha_{13} - \alpha_{12}\right)\Bar{a}b\,, \nonumber
\end{align}
where we have introduced the notation $a := \Bar{z}_1 \cdot z_2 $, $b := \Bar{z}_1 \cdot z_3$ and $c := \Bar{z}_2 \cdot z_3$. Our main goal is to find $a,b,c$. Then, along the lines of the discussion at the end of Section~\ref{SUNSpinChain} the solution to the geodesic equation on $\mathcal{F}(3)$ can as well be constructed.

\subsection{Integrals of motion.} 
Let us write down the Gaudin-inspired integrals of motion for our $\mathrm{SU}(3)$ spin system. They are easily guessed from the e.o.m. (\ref{equatiForCIn}):
\begin{align}
    &\mathsf{h}_1 = (\alpha_{12}-\alpha_{23})|a|^2 + (\alpha_{13}-\alpha_{23})|b|^2\,,\nonumber\\
    &\mathsf{h}_2 = (\alpha_{12} - \alpha_{13})|a|^2 + (\alpha_{23} - \alpha_{13}) |c|^2\,,\label{integralsOfMotion}\\
    &\mathsf{h}_3 = (\alpha_{12} - \alpha_{13})|b|^2 + (\alpha_{12} - \alpha_{23}) |c|^2\,,\nonumber\\
    &\mathsf{h}_4 = \Bar{a}b\Bar{c} + a\Bar{b}c\,. \nonumber
\end{align}
In order to bring the integrals (except $\mathsf{h}_4$) to Gaudin form, one should rescale them. However, in this section we will keep them in their present form.

\subsection{Elliptic curve and the $\wp$-function.}

Now we are ready to solve the equations~(\ref{equatiForCIn}). Firstly, let us derive an equation that contains only $a,b$ or $c$. To this end, we will take the derivative of the equation for~$c$. After some manipulations we arrive~at 
\begin{align}\label{cEquation}
    \frac{\mathrm{d}^2\;\,}{\mathrm{d}\tau^2} c = \left( \frac{\mathsf{h}_2}{\alpha_{13}-\alpha_{23}}+ \frac{\mathsf{h}_3}{\alpha_{23}-\alpha_{12}} + 2 |c|^2\right) c \,, 
\end{align}
where $\tau = \sqrt{(\alpha_{12} - \alpha_{23})(\alpha_{23} - \alpha_{13})}\,t$. Permuting the indices if necessary\footnote{The cases where some of the $\alpha_{ij}$'s are equal are much simpler. For example, if $\alpha_{13} = \alpha_{12}$, then from (\ref{equatiForCIn}) we have $c = \mathrm{const}$, and as a result one has a linear system of ODEs with constant coefficients for $a$ and $b$. This is the case of the so-called submersion  metric~\cite{Bykov_2024}. In the general case of homogeneous spaces metrics of this type were studied in \cite{Souris_2023}.}, one may assume that $\alpha_{12} > \alpha_{23} > \alpha_{13}$. In fact, equations of the same type hold as well for $a$ and $b$:
\begin{align}
    \frac{\mathrm{d}^2\;\,}{\mathrm{d}\tau^2} a = \left(  \frac{\mathsf{h}_3}{\alpha_{23}-\alpha_{12}} + 2 |c|^2\right) a \,,\label{aEquation}\\
    \frac{\mathrm{d}^2\;\,}{\mathrm{d}\tau^2} b = \left(  \frac{\mathsf{h}_2}{\alpha_{13}-\alpha_{23}} + 2 |c|^2\right) b \,.\label{bEquation}
\end{align}
Thus, if we know $|c|^2$,  these are linear equations for $a,b$ and $c$. An equation for $y:=|c|^2$ can in turn be derived from~(\ref{cEquation}):
\begin{align}
    &\frac{\mathrm{d}^2\;\,}{\mathrm{d}\tau^2} y = 6y^2-4\mathrm{I}_2 \,y-4\mathrm{I}_3\, y+2\,\mathrm{I}_2\, \mathrm{I}_3\,,\label{equationY}\\
    &\text{where}\quad \mathrm{I}_2:= - \frac{\mathsf{h}_2}{\alpha_{13} - \alpha_{23}}\,,\;\;\mathrm{I}_3:=- \frac{\mathsf{h}_3}{\alpha_{23} - \alpha_{12}}\,.\nonumber
\end{align}
The obvious first integral of (\ref{equationY}) is
\begin{align}\label{origEqOny}
    \left(\frac{\mathrm{d}\;\,}{\mathrm{d}\tau}y\right)^2 - 4 y \left(y - \mathrm{I}_2\right) \left(y - \mathrm{I}_3\right)= \mathrm{const}=:\mathrm{I}_4\,.
\end{align}
Using the definitions of $\mathrm{I}_2,\mathrm{I}_3$ and (\ref{equatiForCIn}), we may express the constant as 
\begin{align}
    \mathrm{I}_4 = -\frac{\left(\alpha_{12} - \alpha_{13}\right)^2}{\left(\alpha_{12} - \alpha_{23}\right)\left(\alpha_{23} - \alpha_{13}\right)}\,\mathsf{h}_4^2\,.
\end{align}

Equation (\ref{origEqOny}) is an equation defining the Weierstrass $\wp$-function\footnote{We refer to \cite{Akhiezer_1990} and \cite{Pastras_2020} for details on elliptic functions.} as well as the associated elliptic curve. To bring it to standard form, we make the shift $
    {\mathrm{Y} = y - \frac{1}{3}\left(\mathrm{I}_2+\mathrm{I}_3\right)}
$ arriving at
\begin{align}
    &\left(\frac{\mathrm{d}\;\,}{\mathrm{d}\tau} \mathrm{Y}\right)^2 = 4 \mathrm{Y}^3 - g_2 \mathrm{Y} - g_3\,,\label{WeierstrassY}\\
    &\text{where} 
    \quad g_2 = \frac{4}{3}\left(\mathrm{I}_2^2 - \mathrm{I}_2\,\mathrm{I}_3+\mathrm{I}_3^2 \right),\nonumber\\
    & \quad \quad \quad \;\;g_3 = -\mathrm{I}_4 + \frac{4}{27}\left(\mathrm{I}_2+\mathrm{I}_3\right)\left(2\mathrm{I}_2^2 - 5\mathrm{I}_2\,\mathrm{I}_3 + 2 \mathrm{I}_3^2\right)\,.\nonumber
\end{align}
The general solution to equation (\ref{WeierstrassY}) is of the form $\mathrm{Y} = \wp(\tau + \mathcal{O})$, where $\mathcal{O}$ depends on the initial conditions. The invariants $g_2$ and $g_3$ are real; in Appendix \ref{realRoots} we show that the polynomial $4 \mathrm{Y}^3 - g_2 \mathrm{Y} - g_3$ has only real roots. Thus, $\wp$ has one real and one imaginary period, as can be inferred from the theory of elliptic functions. The real period~$2\omega_1$ and the imaginary period~$2\omega_2$ are expressed as \cite{Akhiezer_1990}
\begin{align}
    &\omega_1 = \int^{\infty}_{e_1}\frac{\mathrm{d}\mathrm{Y}}{\sqrt{4 \mathrm{Y}^3-g_2 \mathrm{Y}-g_3}}\,,\\
    &\omega_2 = i\int^{\infty}_{-e_3}\frac{\mathrm{d}\mathrm{Y}}{\sqrt{4 \mathrm{Y}^3-g_2 \mathrm{Y}-g_3}}\,,\nonumber
\end{align}
where $e_1 > e_2 >e_3$ are the roots of $4\mathrm{Y}^3-g_2 \mathrm{Y}-g_3=0$.

The solution to (\ref{WeierstrassY}) must be real and regular because $\mathrm{Y}$ is expressed in terms of~$|c|^2$. Therefore, $\mathcal{O} = -\tau_0 + \omega_2$, where $\tau_0 \in \mathbb{R}$ \cite{Pastras_2020}. As a result, we find that
\begin{align}
    y =|c|^2= \wp\left(\tau - \tau_0 +\omega_2\right) + \frac{1}{3}\left(\mathrm{I}_2+\mathrm{I}_3\right)\,.
\end{align}
We must ensure that $y\geq 0$. This follows from the fact that the value of ${\mathrm{Y} = \wp\left(\tau - \tau_0 +\omega_2\right)}$ lies between the two smallest roots of $4\mathrm{Y}^3 - g_2 \mathrm{Y} - g_3$ \cite{Pastras_2020}. Thus, $y$ lies between the two smallest roots of $4 y \left(y - \mathrm{I}_2\right) \left(y - \mathrm{I}_3\right) + \mathrm{I}_4$, which are non-negative (see Appendix \ref{realRoots}).

We note that $|c|^2$ is periodic with a period of $2\omega_1$, and this is also true for $|a|^2$ and~$|b|^2$, due to the conservation of $\mathsf{h}_2$ and $\mathsf{h}_3$.

\subsection{Lam\'{e} equation and its solution. }

Now we are ready to solve the equations (\ref{cEquation})-(\ref{bEquation}) for $a,b$ and $c$. They read as follows:
\begin{align}
    &\frac{\mathrm{d}^2\,\;}{\mathrm{d}\tau^2}a = \left(\frac{2}{3}\mathrm{I}_2-\frac{1}{3}\mathrm{I}_3+2\wp(\tau - \tau_0 + \omega_2)\right) a\,,\nonumber\\
    &\frac{\mathrm{d}^2\,\;}{\mathrm{d}\tau^2}b = \left(-\frac{1}{3}\mathrm{I}_2+\frac{2}{3}\mathrm{I}_3+2\wp(\tau - \tau_0 + \omega_2)\right) b\,,\label{EqABC}\\
    &\frac{\mathrm{d}^2\,\;}{\mathrm{d}\tau^2}c = \left(-\frac{1}{3}\mathrm{I}_2-\frac{1}{3}\mathrm{I}_3+2\wp(\tau - \tau_0 + \omega_2)\right) c\,.\nonumber
\end{align}
Equations (\ref{EqABC}) are  several copies of the well-known Lam\'{e} equation \cite{Whittaker1996-wv}
\begin{align}
    &\frac{\mathrm{d}^2 w}{\mathrm{d}z^2} = \left(\lambda + n(n+1)\wp(z)\right)w\,,\quad \text{with}\;\,n=1\,.\label{LameEq}
\end{align}
The solutions to (\ref{LameEq}) are \cite{Akhiezer_1990}
\begin{align}\label{w12sols}
    w_{1,2}(z) = e^{\mp z \zeta(\mu)}\frac{\sigma(z\pm \mu)}{\sigma(z)}\,,
\end{align}
where $\mu$ is a root\footnote{In general, there are two roots, $\pm\mu$, corresponding to the two choices of signs in~(\ref{w12sols}).} of the equation $\lambda = \wp(\mu)$ and $\zeta, \sigma$ are the respective Weierstrass functions corresponding to $\wp(z)$. In our case $z = \Delta \tau +\omega_2:= \tau -\tau_0 +\omega_2$. Therefore, it is convenient to use
\begin{align}
    y_{\pm}(\tau | \mu) = e^{\mp\tau \zeta( \mu)}\frac{\sigma(\Delta\tau+\omega_2\pm \mu)\sigma(-\tau_0 + \omega_2)}{\sigma(\Delta\tau+\omega_2)\sigma(-\tau_0 +  \omega_2\pm \mu)}
\end{align}
because $y_{\pm}(0 | \mu)  = 1$. Thus, the solutions for $a,b$ and $c$ take the form
\begin{align}
    &a = \mathcal{A}_{+}\,y_{+}(\tau|\mu_1) + \mathcal{A}_{-}\,y_{-}(\tau|\mu_1)\,,\nonumber\\
    &b = \mathcal{B}_{+}\,y_{+}(\tau|\mu_2) + \mathcal{B}_{-}\,y_{-}(\tau|\mu_2)\,,\label{ABCsolutions}\\
    &c = \mathcal{C}_{+}\,y_{+}(\tau|\mu_3) + \mathcal{C}_{-}\,y_{-}(\tau|\mu_3)\,,\nonumber
\end{align}
where $\wp(\mu_1) = 2\mathrm{I}_2/3-\mathrm{I}_3/3$, $\wp(\mu_2) = -\mathrm{I}_2/3+2\mathrm{I}_3/3$, $\wp(\mu_3) = -\mathrm{I}_2/3-\mathrm{I}_3/3$ and $\mathcal{A}_{\pm},\mathcal{B}_{\pm},\mathcal{C}_{\pm}$ are some constants to be determined. In Appendix \ref{boundSolution} we prove that the solutions are bounded. 

To find the constants, we will use the original equations (\ref{equatiForCIn}) at $\tau=0$ together with the initial conditions for $a,b$ and $c$:
$a_0 = a(0),\,b_0 = b(0)$ and $c_0 = c(0)$.
Using the addition theorem\footnote{It reads $\zeta(u+v)-\zeta(u)-\zeta(v) = \displaystyle\frac{1}{2}\frac{\wp'(u) - \wp'(v)}{\wp(u)-\wp(v)}$, cf.~\cite{Akhiezer_1990}.} for $\zeta$, one easily finds that
\begin{align}
    &\mathcal{A}_{\pm} = \frac{1}{2}\mathrm{M}_{\pm}(\mu_1)a_0 \pm i\,\mathrm{N}(\mu_1) \sqrt{\frac{\alpha_{23}-\alpha_{13}}{\alpha_{12}-\alpha_{23}}}b_0 \Bar{c}_0\,,\nonumber\\
    &\mathcal{B}_{\pm} = \frac{1}{2}\mathrm{M}_{\pm}(\mu_2)b_0 \mp i\,\mathrm{N}(\mu_2) \sqrt{\frac{\alpha_{12}-\alpha_{23}}{\alpha_{23}-\alpha_{13}}}a_0 c_0\,,\label{initialData}\\
    &\mathcal{C}_{\pm} = \frac{1}{2}\mathrm{M}_{\pm}(\mu_3)c_0 \mp i\,\mathrm{N}(\mu_3)\frac{\alpha_{12}-\alpha_{13}}{\sqrt{(\alpha_{12}-\alpha_{23})(\alpha_{23}-\alpha_{13})}}\Bar{a}_0 b_0\,,\nonumber\\
    &\text{where}\quad \mathrm{M}_{\pm}(\mu):=\left(1\pm\frac{\wp'(-\tau_0+\omega_2)}{\wp'(\mu)}\right)\,,\quad \mathrm{N}(\mu):=\frac{\wp(-\tau_0+\omega_2)-\wp(\mu)}{\wp'(\mu)}\,.\nonumber
\end{align}

This completes the solution of equations~(\ref{equatiForCIn}). Using the algorithm of~(\ref{Zformalsol})-(\ref{Vsol}) one can then find the evolution of vectors $u_{1}, u_2, u_3$ parametrizing $\mathcal{F}(3)$.

Summarizing the above, we arrive at the following proposition:
\begin{prop}
    The solution to the Euler equation (\ref{equatiForCIn}) is given by (\ref{ABCsolutions}) and (\ref{initialData}). The respective geodesic can be found using (\ref{Vsol}) and (\ref{geodUsol}).
\end{prop}

\subsection{Special solutions.} 
In this section we will describe two types of special solutions. To this end let us study more closely the elliptic curve $\mathrm{Z}^2=4 \mathrm{Y}^3 - g_2 \mathrm{Y} - g_3$ entering the conservation law~(\ref{WeierstrassY}). Recall that earlier we denoted by $e_1, e_2, e_3$ the roots of the cubic polynomial in the r.h.s. of this equation. Solutions of special type arise when the elliptic curve degenerates, i.e. when some of the roots coincide. 

We note that $\wp(\tau-\tau_0+\omega_2) \in [e_3,e_2]$. When $e_2 = e_3$, $\wp(\tau-\tau_0+\omega_2)$ is constant and therefore so is $|c|^2$, as well as $|a|^2, |b|^2$ since they can be expressed through $|c|^2$ via the integrals of motion. 
In this case we have the following algebraic relations between $a,b$ and $c$ (see Appendix \ref{realRoots} for the derivation)
\begin{align}
    &\left(\alpha_{12}-\alpha_{13}\right)|a|^2|b|^2+\left(\alpha_{23}-\alpha_{12}\right)|a|^2|c|^2 + \left(\alpha_{13}-\alpha_{23}\right)|b|^2|c|^2=0\,,\label{conditionSpecial}\\
    &\Bar{a}b\Bar{c}= a\Bar{b}c\,. \label{phaseConOriginal}
\end{align}

\subsubsection{Geodesics in $\CP^1$.} 
One obvious solution is when two of the variables $a,b,c$ are zero. For example, consider $b=c=0$. Then, from (\ref{equatiForCIn}), we find that $a$ is constant. One then easily writes out the solution to the geodesic equation:
\begin{align}
    \mathsf{U}(t)=\mathrm{U}(0)\times \exp\left[-i\begin{pmatrix}
        0 & \alpha_{12}\,a & 0 \\
        \alpha_{12}\,\Bar{a} & 0 & 0 \\
        0 & 0 & 0 
    \end{pmatrix}t\right]\,.
\end{align}
Clearly, this is a geodesic in the fiber $\mathbb{CP}^1$ of the forgetful bundle\footnote{We refer to \cite{Manin1988GaugeFT} for a definition of forgetful bundle and to \cite{Onischik1993} for general information on bundles over homogeneous spaces.} $\mathcal{F}(3) \xrightarrow{\mathbb{CP}^1} \mathbb{CP}^2$, where the projection maps $\mathsf{U} = \begin{pmatrix}
    u_1 & u_2 & u_3
\end{pmatrix}$ to $u_3$.

\subsubsection{Other special solutions.} 
Another solution corresponds to a situation where all of $a,b,c$ are non-zero. Let us look separately at the amplitudes and phases: $a=|a|\,e^{i\psi_a}, b=|b|\,e^{i\psi_b}$ and $c=|c|\,e^{i\psi_c}$. 
Condition~(\ref{conditionSpecial}) then relates $|a|, |b|, |c|$, whereas~(\ref{phaseConOriginal}) leads to
\begin{align}
    &\psi_a - \psi_b + \psi_c \equiv 0 \quad (\text{mod}\;\,\pi)\,.\label{phaseCond}
\end{align}
In this case equations~(\ref{equatiForCIn}) imply that $|a|, |b|, |c|$ are constant (as already discussed above), whereas the phases of $a, b, c$ satisfy the following equations:
\begin{align}
    \dot{\psi}_a = \left(\alpha_{13} - \alpha_{23}\right)\frac{|b||c|}{|a|}\,,\quad 
    \dot{\psi}_b = \left(\alpha_{12} - \alpha_{23}\right)\frac{|a||c|}{|b|}\,,\quad 
    \dot{\psi}_c = \left(\alpha_{12} - \alpha_{13}\right)\frac{|a||b|}{|c|}\,.\label{psieq}
\end{align}
One easily checks that the condition $\dot{\psi}_a - \dot{\psi}_b + \dot{\psi}_c = 0$, which is a consequence of~(\ref{phaseCond}), is satisfied thanks to the constraint~(\ref{conditionSpecial}) on the absolute values of $a, b, c$. It follows from~(\ref{psieq}) that the phases are linear functions of time.

Next, we wish to solve~(\ref{ZeqNewForm}). Using (\ref{phaseCond}), one finds $\psi_b = \psi_a + \psi_c + k\pi$, where $k \in \mathbb{Z}$. This allows rewriting~(\ref{ZeqNewForm}) as 
\begin{align}
    &i\, \frac{\mathrm{d}\,\;}{\mathrm{d}t}\mathsf{Z}= \mathsf{Z}\times\mathcal{D}(t)\,\mathrm{L}\,\mathcal{D}^{-1}(t)\,,\\
    &\text{where}\;\;\mathcal{D}(t):=\begin{pmatrix}
        e^{i\psi_a(t)} & 0 & 0 \\
        0 & 1 & 0\\
        0 & 0 & e^{-i\psi_c(t)}
    \end{pmatrix}\,,\quad \mathrm{L}:=\begin{pmatrix}
        0 & \alpha_{12}|a| & \pm\alpha_{13}|b| \\
        \alpha_{12}|a| & 0 & \alpha_{23}|c| \\
        \pm\alpha_{13}|b| & \alpha_{23}|c| & 0
    \end{pmatrix}\,,\nonumber
\end{align}
and $\pm$ depends on the value of $k$: $\pm=(-1)^{k}$. 
 To solve the equation, we make the change $\mathsf{G} = \mathsf{Z}\mathcal{D}$, which is  a gauge transformation at the same time. Ultimately we arrive at the following equation for $\mathsf{G}$:
\begin{align}
    &i\, \frac{\mathrm{d}\,\;}{\mathrm{d}t}\mathsf{G}=\mathsf{G}\times \mathcal{K}\,,\\
    &\text{where}\;\;
    \mathcal{K}:=
    \begin{pmatrix}
        \left(\alpha_{23} - \alpha_{13}\right)\frac{|b||c|}{|a|} & \alpha_{12}|a| & \pm\alpha_{13}|b| \\
        \alpha_{12}|a| & 0 & \alpha_{23}|c| \\
        \pm\alpha_{13}|b| & \alpha_{23}|c| & \left(\alpha_{12} - \alpha_{13}\right)\frac{|a||b|}{|c|}
    \end{pmatrix}\,. \label{Kmatrix}
\end{align}
Since $\mathcal{K}$ is a constant matrix, this is solved immediately. Upon finding $\mathsf{G}$, we then obtain a geodesic on $\mathcal{F}(3)$ via 
\begin{align}
    \mathsf{U}(t) =  \mathsf{G}\left(\mathsf{G}^{\dagger}\mathsf{G}\right)^{-1/2}=\mathsf{U}(0)\,e^{-i\,\mathcal{K}\,t}\,.
\end{align}
We shall emphasize once again that this is a solution as long as the condition (\ref{conditionSpecial}) holds and the initial values of the phases satisfy $\psi^0_a - \psi^0_b + \psi^0_c \equiv 0 $ $ (\text{mod}\;\,\pi)$. Besides, notice that the solutions with different choices of signs in~(\ref{Kmatrix}) are not essentially different. Indeed, one can obtain one from the other by the gauge transformation  $\mathsf{G}\mapsto \mathsf{G}\cdot \mathrm{Diag}(-1,1,-1)$ with a simultaneous swap $t\mapsto -t$. Thus, they  correspond to the same geodesic running in opposite directions. 

The meaning of this solution is that, even in the case of the most general metric, a family of geodesics that are orbits of one-parametric subgroups in $\mathrm{U}(3)$ still exists. Notice, however, that the subgroup is deformed compared to the case of the normal metric, where $\alpha_{12}=\alpha_{13}=\alpha_{23}$.

\section{Spectrum of the Laplacian on $\mathcal{F}(3)$}\label{F3specsec}
In this section we will study the spectrum of the Laplace-Beltrami operator on~$\mathcal{F}(3)$. It follows from  Proposition \ref{spinFlagConnection}  that the spectrum for the metric (\ref{F3metric}) can be recovered in the limit $p\to \infty$ as a spectrum of the spin chain Hamiltonian
\begin{align}\label{spinHam}
    \mathcal{H} = \alpha_{12} \mathrm{S}_{1}^a \mathrm{S}_{2}^a+\alpha_{13} \mathrm{S}_{1}^a \mathrm{S}_{3}^a+\alpha_{23} \mathrm{S}_{2}^a \mathrm{S}_{3}^a +\mathrm{const}\,.
\end{align}
As we noted in Section \ref{geodSection}, the Hamiltonian (\ref{spinHam}) is a linear combination of Gaudin Hamiltonians. Thus, its spectrum can be obtained using the Bethe ansatz method~\cite{Feigin1994-ag}. Here we will carry out the relevant calculation.

Before proceeding to the calculation, let us discuss some general properties of the spin chain we are dealing with. First, the appropriate Gaudin Hamiltonians can be easily derived from the classical integrals of motion~(\ref{integralsOfMotion}), using the correspondence\footnote{This follows from the Schwinger-Wigner quantization scheme \cite{Affleck_2022}, which we will recall shortly.} between classical and quantum observables ${|\Bar{z}_1\cdot z_2|^2=|a|^2 \mapsto ~\mathrm{S}_1^a \mathrm{S}_2^a,}$\,${|\Bar{z}_1\cdot z_3|^2=|b|^2 \mapsto \mathrm{S}_1^a \mathrm{S}_3^a,}$\,$|\Bar{z}_2\cdot z_3|^2=|c|^2 \mapsto \mathrm{S}_2^a \mathrm{S}_3^a$ :
\begin{align}
    -\frac{\mathsf{h}_1}{(\alpha_{23}-\alpha_{12})(\alpha_{23} - \alpha_{13})}\;\mapsto &\;\;\mathrm{H}_1=\frac{\mathrm{S}_1^a \mathrm{S}_2^a}{\alpha_{23} - \alpha_{13}} + \frac{\mathrm{S}_1^a \mathrm{S}_3^a}{\alpha_{23} - \alpha_{12}}\,, \label{GaudinHam1}\\
    -\frac{\mathsf{h}_2}{(\alpha_{13}-\alpha_{12})(\alpha_{13} - \alpha_{23})}\;\mapsto &\;\;\mathrm{H}_2 = \frac{\mathrm{S}_1^a \mathrm{S}_2^a}{\alpha_{13} - \alpha_{23}} + \frac{\mathrm{S}_2^a \mathrm{S}_3^a}{\alpha_{13}-\alpha_{12}}\,,\label{GaudinHam2}\\
    \frac{\mathsf{h}_3}{(\alpha_{12}-\alpha_{13})(\alpha_{12} - \alpha_{23})}\;\mapsto&\;\;\mathrm{H}_3 =\frac{\mathrm{S}_1^a \mathrm{S}_3^a}{\alpha_{12}-\alpha_{23}}+\frac{\mathrm{S}_2^a \mathrm{S}_3^a}{\alpha_{12}-\alpha_{13}}\,.\label{GaudinHam3}
\end{align}
One therefore reads off the Gaudin parameters (see (\ref{gaudinHam})) as
\bea
z_1 = \alpha_{23},\quad z_2 =\alpha_{13}, \quad z_3 = \alpha_{12}\,.
\eea
Using the $\mathrm{H}_i$'s, one may write the spin Hamiltonian (\ref{spinHam}) as follows\footnote{The parameters $t_i$ entering~(\ref{tham}) can be easily read off from this formula using the fact that $\mathcal{C}_2=2\left(\alpha_{23} \mathrm{H}_1+\alpha_{13} \mathrm{H}_2+\alpha_{12} \mathrm{H}_3\right)+\mathrm{const.}$}: 
\begin{align}
    \mathcal{H} = \frac{\alpha_{12}+\alpha_{23}+\alpha_{13}}{2}\mathcal{C}_2 - \left(\alpha^2_{23}\mathrm{H}_1 + \alpha^2_{13}\mathrm{H}_2 + \alpha^2_{12}\mathrm{H}_3\right)+\mathrm{const}\,,\label{fullHam}
\end{align}
where $\mathcal{C}_2 = (\mathrm{S}_1^a + \mathrm{S}_2^a + \mathrm{S}_3^a)^2$ is a quadratic Casimir operator 
acting in $\mathrm{Sym}(p)^{\otimes 3}$. We shall also choose the additive constant by the requirement that the eigenvalue of $\mathcal{H}$, corresponding to the trivial representation in $\mathrm{Hilb}(p)$, is zero. This is also the ground state energy of the Hamiltonian corresponding to constant functions on $\mathcal{F}(3)$.

Now, let us turn to the study of the Hilbert space of the spin chain. As we know from Proposition \ref{spinFlagConnection}, $\lim_{p\to\infty}\mathrm{Hilb}(p) = \mathrm{L}^2(\mathcal{F}(3))$ in terms of $\mathrm{SU}(n)$ representations. In fact, the Hilbert spaces corresponding to different values of $p$ are simply embedded in each other:
\bea
\mathrm{Hilb}(p)~\hookrightarrow~\mathrm{Hilb}(p+1)\,. \label{Hilbembed}
\eea
Let us describe the embedding in terms of Young diagrams corresponding to the representations in $\mathrm{Hilb}(p)$ and $\mathrm{Hilb}(p+1)$. The embedding~(\ref{Hilbembed}) maps a diagram from $\mathrm{Hilb}(p)$ to a diagram with an additional column of three boxes in $\mathrm{Hilb}(p+1)$. These two diagrams correspond to the same irreducible representation of $\mathrm{SU}(3)$. Summarizing, we conclude that $\mathrm{Hilb}(p) \hookrightarrow \mathrm{L}^2(\mathcal{F}(3))$.

\subsection{Functions on $\mathcal{F}(3)$ and  explicit computation of the spectrum.}\label{compSpec}
In this section we will explicitly describe the embedding of $\mathrm{Hilb}(p)$ into $\mathrm{L}^2(\mathcal{F}(3))$ and show how it can be used to find the eigenvalues of $\mathcal{H}$, first by a direct computer calculation.

We recall that the spin chain can be realized in terms of the so-called Schwinger-Wigner oscillators~\cite{Affleck_2022}. To this end, we define the set of usual bosonic creation-annihilation operators $\{a_m,a_m^{\dagger}, b_m,b_m^{\dagger},c_m,c_m^{\dagger}\}_{m=1}^3$ (each letter for one of the $\mathrm{Sym}(p)$ factors in $\mathrm{Hilb}(p)$, see (\ref{Hilbp})) with the following non-trivial commutation relations
\begin{align}
    \left[a_m, a^{\dagger}_l\right] = \left[b_m, b^{\dagger}_l\right] = \left[c_m, c^{\dagger}_l\right] =\delta_{ml}\,.
\end{align}
We also define the Fock space vacuum vector $\ket{0}$ such that $a_m\ket{0}=b_m\ket{0}=c_m\ket{0}=0$. The operators $\mathrm{S}^{a}_i$ are realized as 
\begin{align}\label{generatorsRep}
    \mathrm{S}^{a}_1 = a_i^{\dagger}\tau^a_{ij}a_j, \quad \mathrm{S}^{a}_2 = b_i^{\dagger}\tau^a_{ij}b_j, \quad \mathrm{S}^{a}_3 = c_i^{\dagger}\tau^a_{ij}c_j,
\end{align}
where $\tau^a$ are $\mathrm{SU}(3)$ generators in the defining representation. In these terms, $\mathrm{Hilb}(p)$ is the space of homogeneous polynomials in $a^{\dagger}, b^{\dagger}, c^{\dagger}$ of poly-degree $(p,p,p)$ acting on the vacuum vector $\ket{0}$. 

Now we are in a position to explicitly describe the embedding $\mathrm{Hilb}(p)\hookrightarrow \mathrm{L}^2(\mathcal{F}(3))$. Let us take a state $\mathrm{Q}(a^{\dagger},b^{\dagger},c^{\dagger})\ket{0} \in \mathrm{Hilb}(p)$, where $\mathrm{Q}$ is a homogeneous polynomial. We map this state to the function
\begin{align}
    f(x,y,z) = \frac{\mathrm{Q}(x,y,z)}{\langle x, y, z\rangle^p}\,,
\end{align}
where $x,y,z \in \mathbb{CP}^2$, $\langle x, y, z\rangle = \epsilon_{ijk}x^{i}y^{j}z^{k}$, $\epsilon_{ijk}$ is a totally antisymmetric tensor and $\epsilon_{123}=1$. This is a well-defined function\footnote{More exactly, this is true outside the divisor where the denominator vanishes, i.e. we assume $\langle x, y, z\rangle\neq 0$.} on $(\mathbb{CP}^2)^{\times 3}$ because $\mathrm{Q}$ is a homogeneous polynomial of a proper poly-degree. Finally, to get a function on the flag manifold we restrict $f(x,y,z)$ to $\mathcal{F}(3) \hookrightarrow (\mathbb{CP}^2)^{\times 3}$ -- the embedding is defined by the requirement that $x,y,z$ are three orthogonal lines in $\mathbb{C}^3$. 
We note that the map $\mathrm{Q}\ket{0}\mapsto f|_{\mathcal{F}(3)}$ is injective, since $\mathrm{Q}$ is a polynomial, and if it vanishes on a continuous set, such as $\mathcal{F}(3)$, then it must be identically zero. In these terms we can as well describe the embedding $\mathrm{Hilb}(p) \hookrightarrow \mathrm{Hilb}(p+1)$ of~(\ref{Hilbembed}): one simply needs to multiply $\mathrm{Q}\ket{0} \in \mathrm{Hilb}(p)$ by $\langle a^{\dagger},b^{\dagger},c^{\dagger}\rangle$ to find the image of the state in $\mathrm{Hilb}(p+1)$. 

Next, using (\ref{generatorsRep}), one finds that 
\begin{align}
    \mathrm{S}^{a}_1 \mathrm{S}^{a}_2 = \left(a_i^{\dagger} b_i\right) \left(b_j^{\dagger} a_j \right) + \mathrm{Const}(p)\,.
\end{align}
Here we have made use of the identity $\sum_{a=1}^{8} \tau^a\otimes \tau^a = \mathrm{P}-\frac{1}{3}\mathrm{Id}$, where $\mathrm{P}$ is a permutation, and $\mathrm{Const}(p)$ is a constant that depends on $p$. For the moment we are not paying much attention to constant terms since there is an overall additive constant in the definition of $\mathcal{H}$. Analogous equalities hold for $\mathrm{S}^{a}_1 \mathrm{S}^{a}_3$ and $\mathrm{S}^{a}_2 \mathrm{S}^{a}_3$.  One easily verifies that 
\begin{align}
    \left[\mathrm{S}^{a}_i \mathrm{S}^{a}_j - \mathrm{Const}(p), \langle a^{\dagger},b^{\dagger},c^{\dagger}\rangle\right] = 0\,. \label{spinspinoscillators}
\end{align}
Therefore, following the explicit description of the embedding $\mathrm{Hilb}(p) \hookrightarrow \mathrm{Hilb}(p+1)$, we conclude that $\mathrm{Spec}(p)$, the spectrum of $\mathcal{H}$ in $\mathrm{Hilb}(p)$,  is contained in $\mathrm{Spec}(p+1)$. This is an additional illustration of Proposition \ref{spinFlagConnection} that explains why the limit $\lim_{p\to \infty} \mathrm{Spec}(p)$ exists.  

The above representation~(\ref{spinspinoscillators}) in terms of the creation-annihilation operators is also useful for the analysis of the full Hamiltonian~(\ref{spinHam}). Our requirement throughout this paper is that the additive constant in~(\ref{spinHam}) should be chosen in such a way that the ground state eigenvalue is zero. The so-normalized Hamiltonian should take the following form:
\bea\label{HamABC}
\mathcal{H}=\alpha_{12}\,\left(a_i^{\dagger} b_i\right) \left(b_j^{\dagger} a_j \right)+\alpha_{23} \left(b_i^{\dagger} c_i\right) \left(c_j^{\dagger} b_j \right)+\alpha_{13} \left(a_i^{\dagger} c_i\right) \left(c_j^{\dagger} a_j \right)
\eea
First of all, it is easy to see that this Hamiltonian is non-negative: $\mathcal{H}\geq 0$. Besides, the zero-energy ground state is easily constructed:
\bea
|\Psi_0\rangle=\left(\langle a^{\dagger},b^{\dagger},c^{\dagger}\rangle\right)^p |0\rangle \,,\quad\quad \mathcal{H}|\Psi_0\rangle=0\,.
\eea
Certain other eigenstates can also be found almost immediately. Indeed, consider the state
\bea
|\Psi_{3p}\rangle=\Psi_{i_1\cdots i_p\,j_1 \cdots j_p\,k_1\cdots k_p} a_{i_1}^\dagger\cdots a_{i_p}^\dagger\,b_{j_1}^\dagger\cdots b_{j_p}^\dagger\,c_{k_1}^\dagger\cdots c_{k_p}^\dagger |0\rangle 
\eea
where the wave function $\Psi_{\cdots}$ is \emph{fully symmetric} w.r.t. all indices. Clearly, this corresponds to the representation
\bea
\underbrace{
        \begin{ytableau}
               *(White)~ & *(White)~ & *(White)~ & *(White)~&*(White)~&*(White)~
        \end{ytableau}
        }_{3p}
\eea
Now, acting on the state $|\Psi_{3p}\rangle$ by the Hamiltonian~(\ref{HamABC}), one readily finds the eigenvalue:
\bea\label{3prepeigenvalue}
\mathcal{H}|\Psi_{3p}\rangle=\left(\alpha_{12}+\alpha_{23}+\alpha_{13}\right)\,p(p+1)|\Psi_{3p}\rangle
\eea

We also note that $\mathrm{S}^a_i \mathrm{S}_j^a$ commutes with the diagonal action of $\mathrm{SU}(3)$ in $\mathrm{Hilb}(p)$, whose generators have the form $\mathrm{S}^a = \mathrm{S}^a_1+\mathrm{S}^a_2+\mathrm{S}_3^a$. Therefore, the Gaudin Hamiltonians $\mathrm{H}_i$'s and, consequently, $\mathcal{H}$ are constant on irreducible representations of $\mathrm{SU}(3)$ in $\mathrm{Hilb}(p)$. The problem with diagonalization only arises for those irreducible representations which appear  in $\mathrm{Hilb}(p)$ with multiplicity. For example, consider $p=1$, where one has the following decomposition into irreducibles 
\begin{align}
    \mathrm{Hilb}(1) 
        =
        \begin{ytableau}
               ~ & ~ & ~
        \end{ytableau}
        \oplus
        2\;
        \begin{ytableau}
               ~ & ~ \\
               ~ &\none 
        \end{ytableau}
        \oplus
        \begin{ytableau}
               ~ \\
              ~ \\
               ~
        \end{ytableau}\,.
\end{align}
As we see, the irreducible representation $\tiny{\begin{ytableau}
               ~ & ~ \\
               ~ &\none 
        \end{ytableau}}$ appears twice.

\vspace{0.1cm}
        The characteristic polynomial of $\mathcal{H}$ factors into a product of polynomials for each irreducible representation in $\mathrm{Irr}\left(\mathrm{Hilb}(p)\right)$ (the set of all irreducible representations in $\mathrm{Hilb}(p)$), the degree of each factor being equal to the multiplicity of that representation. Each polynomial is exponentiated to a power equal to  the dimension of the corresponding representation. Our main goal is to find these polynomials.

Let us analyze what irreducible representations appear in $\mathrm{Hilb}(p)$ and, consequently, in $\mathrm{L}^2\left(\mathcal{F}(3)\right)$. For convenience, we introduce a special notation for irreducible representations, namely we refer to them as type $(n, m)$ if they have the following Young diagram: 
\begin{align}\label{nmdiagram}
    \begin{overpic}[scale=0.2,unit=1mm]{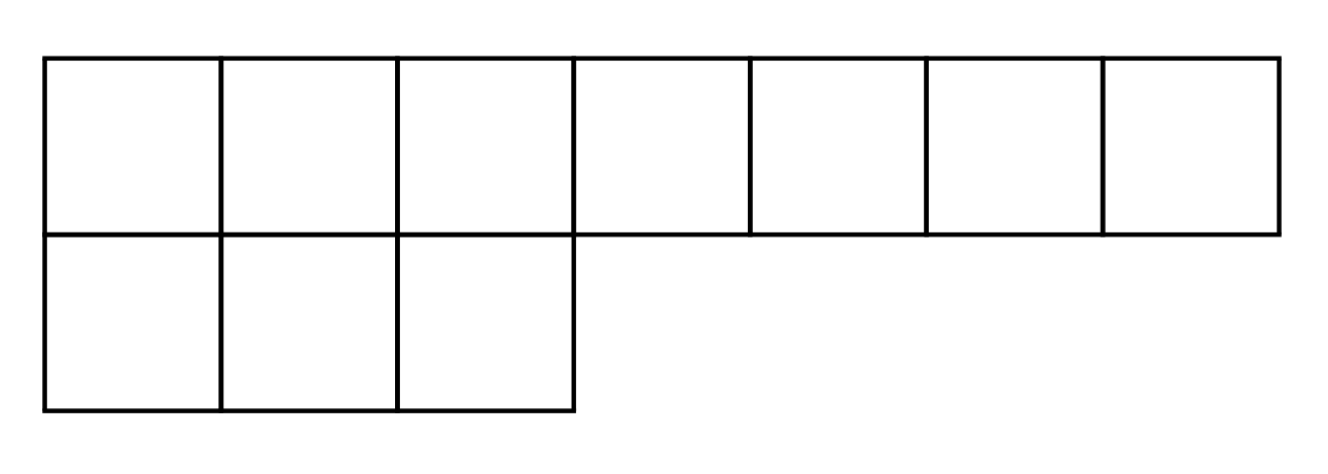}
    \put(20.5,6){$\xleftrightarrow{\makebox[2.2cm]{}}$}
    \put(32,4){$n$}
    \put(2,-0.5){$\xleftrightarrow{\makebox[1.6cm]{}}$}
    \put(9,-2.5){$m$}
    \end{overpic}
\end{align}
In general, a Young diagram may contain columns of three boxes, but these do not affect the irreducible representation type because they correspond to the trivial representation of $\mathrm{SU}(3)$\footnote{We refer to \cite{fulton1991representation} for general information on the correspondence between irreducible representations and Young diagrams.}. An important property is that each representation will typically enter with a certain multiplicity that may be different from zero or one, since $\mathcal{F}(3)$ is not a multiplicity free space. This is what makes it different from the case of symmetric spaces, where each representation  appears with multiplicity at most one (cf.~\cite{Kramer1979, GuilleminSternberg}).  

As we noticed above, $\mathrm{Hilb}(p)\hookrightarrow \mathrm{Hilb}(p+1)$. Therefore, to find irreducible representations, at each value of $p$ it suffices to look for those that are in $\mathrm{Hilb}(p+1) \backslash \mathrm{Hilb}(p)$\footnote{$\mathrm{Hilb}(0)$ contains only the trivial representation.}. In fact, they have Young diagrams with two rows. If not, then the Young diagram contains columns of three boxes. One can simply erase them and find that the remaining diagram is in $\mathrm{Hilb}(p)$, which follows from the multiplication rule for Young diagrams. Thus, we arrive at the following description of  Young diagrams in $\mathrm{Hilb}(p+1) \backslash \mathrm{Hilb}(p)$:
\bea\label{nmpstates}
(n,m)\in \mathrm{Hilb}(p+1) \Big\backslash \mathrm{Hilb}(p)\quad  \Longleftrightarrow \quad   n+2m=3(p+1)\,.
\eea
Let us illustrate the way these diagrams arise in $\mathrm{Hilb}(p+1)$:
\begin{align}
    \mathrm{Hilb}(p+1) &:= \underbrace{
        \begin{ytableau}
               *(BurntOrange)~ & *(BurntOrange)~ & *(BurntOrange)~ & *(BurntOrange)~
        \end{ytableau}
        }_{p+1} \otimes \underbrace{
        \begin{ytableau}
               *(Emerald)~ & *(Emerald)~ & *(Emerald)~ & *(Emerald)~
        \end{ytableau}
        }_{p+1} \otimes \underbrace{
        \begin{ytableau}
               *(RedViolet)~ & *(RedViolet)~ & *(RedViolet)~ & *(RedViolet)~
        \end{ytableau}
        }_{p+1} = \nonumber\\
        &= \left(\bigoplus_{l=0}^{p+1} \underset{\xleftrightarrow{\;\quad l\quad\;}\qquad\;\qquad\qquad}{\begin{ytableau}
            *(BurntOrange)~ & *(BurntOrange)~ & *(BurntOrange)~ & *(BurntOrange)~ & *(Emerald)~ & *(Emerald)~  \\
           *(Emerald)~ & *(Emerald)~ & \none & \none & \none & \none
        \end{ytableau}} \right)\otimes \begin{ytableau}
               *(RedViolet)~ & *(RedViolet)~ & *(RedViolet)~ & *(RedViolet)~
        \end{ytableau} = \label{YoungTables}\\
        &= \bigoplus_{l=0}^{p+1}\;\bigoplus_{k=0}^{\mathrm{min}(2(p+1)-2l;p+1)} \underset{\xleftrightarrow{\;\quad l\quad\;}\xleftrightarrow{\,\;\;\;\;\quad k\quad\;\;\;\,}\qquad\;\;\;\;\,\,}{\begin{ytableau}
            *(BurntOrange)~ & *(BurntOrange)~ & *(BurntOrange)~ & *(BurntOrange)~ & *(Emerald)~ & *(Emerald)~ & *(RedViolet)~ \\
            *(Emerald)~ & *(Emerald)~ & *(RedViolet)~ & *(RedViolet)~ & *(RedViolet)~ & \none & \none
        \end{ytableau}}\;\; \oplus \dots\nonumber\;
\end{align}
The ellipsis stands for irreducible representations which appear in $\mathrm{Hilb}(p)$. Columns with boxes of the same color are prohibited, which is why we have added a restriction on the range of $k$. We learn that all representations with $n+2m=3(p+1)$ for $n,m \in \mathbb{N}_0$ are allowed.

\vspace{1cm}
\begin{lem}
    The multiplicity of the $(n, m)$ representation is $\mathrm{min}(m,n)+1$.
\end{lem}

\begin{proof}
    Notice that the multiplicity is equal to the number of  ways to color (in purple and green, see (\ref{YoungTables})) the second row of the $(n,m)$ diagram. Suppose\footnote{The case $n\geq m$ is considered in complete analogy.} $m\geq n$, then in (\ref{YoungTables}) there is a coloring of the $(n,m)$ diagram without purple boxes in the first row. Indeed, it follows from $n+2m=3(p+1)$ and $m\geq n$ that $m\geq p+1$, so that there are enough slots in the second row that can be filled with all of the purple boxes.  
    Other colorings of $(n,m)$ can be obtained from this  diagram by exchanging part of the purple boxes in the second row with the green boxes in the first one, under the condition that blocks of the same color do not appear in the same column (this is one of the multiplication rules for Young diagrams). Thus, we obtain $n+1$ different colorings. 
\end{proof}

We combine all of the information about irreducible representations in $\mathrm{L}^2(\mathcal{F}(3))$ in a single  table\footnote{In fact, the decomposition into irreducibles (except for multiplicities) for $\mathrm{L}^2(\mathcal{F}(3))$ was found in~\cite{Yamaguchi}.}:
\begin{figure}[H]
    \centering
    \includegraphics[width=0.5\textwidth]{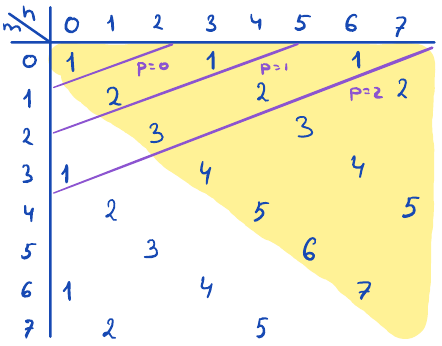}
    \caption{Decomposition of $\mathrm{L}^2(\mathcal{F}(3))$  into irreducible representations. Integers inside the table stand for multiplicities of the corresponding representations.}
    \label{table1}
\end{figure}
Here, in each cell, we have specified the multiplicity of the irreducible representation of type $(n, m)$ (the multiplicity is zero if not specified). We have separated the representations that fall into $\mathrm{Hilb}(0)$, $\mathrm{Hilb}(1)$, and $\mathrm{Hilb}(2)$, using purple lines. In general, the representations that are contained in $\mathrm{Hilb}(p)$ are those for which $n+2m = 3q$ for $q=0,1,\dots,p$. The yellow shading indicates representations with $n\geq m$. 

We see that $\lim_{p\to \infty}\mathrm{Hilb}(p)$ contains each representation $(n,m)$ along with its conjugate  $(m,n)$ with the same multiplicity. The Laplace-Beltrami operator is real, so that if $\varphi$ is a (complex) eigenfunction transforming in a given representation $(n,m)$, then $\varphi^\ast$ is an eigenfunction in the dual representation $(m,n)$, moreover the eigenvalues are the same. From now on, we will only study the spectrum for representations with $n\geq m$.

Let us now determine explicitly the spectrum $\mathrm{Spec}(p)$ for small values of $p$, which is thus part of the spectrum for a free particle on the space $\mathcal{F}(3)$ equipped with the metric (\ref{F3metric}). To this end it will be useful to introduce the elementary $\mathrm{S}_3$ symmetric polynomials of $\alpha_{12}, \alpha_{23}, \alpha_{13}$:
\bea\label{s1s2s3}
s_{1} := \alpha_{12} + \alpha_{23} + \alpha_{13},\quad  s_2:= \alpha_{12}\alpha_{13} + \alpha_{12}\alpha_{23} + \alpha_{13}\alpha_{23},\quad  s_3 := \alpha_{12}\alpha_{23}\alpha_{13}\,.
\eea

To calculate $\mathrm{Spec}(p)$, we use the bosonic representation of the spin system described above. In the case of $p=1$ and $p=2$, we find \cite{OscCalculusCoadj}:

\begin{align*}
&&& \textrm{Eigenvalue}\, &&&&\textrm{Representation} && \textrm{Dimension}\nonumber\\
    & p = 1\,:\nonumber\\
    &&&\lambda = 0 &&&&Trivial && 1 \nonumber \\
    &&&\lambda - 2 s_1 =0 &&&&{\tiny \begin{ytableau}
        ~ & ~ & ~
    \end{ytableau}} && 10 \nonumber \\
    &&&\lambda^2 - 2 s_1 \lambda + 3s_2 = 0  &&&&{\tiny \begin{ytableau}
        ~ & ~ \\
        ~ 
    \end{ytableau}} && 8 \\
    & p=2\,:\nonumber\\
    &&&\lambda = 0  &&&&Trivial && 1 \nonumber \\
    &&&\lambda - 2 s_1 =0  &&&&{\tiny \begin{ytableau}
        ~ & ~ & ~
    \end{ytableau}} && 10 \nonumber \\
    &&&\lambda - 2 s_1 =0  &&&&{\tiny \begin{ytableau}
        ~ & ~ & ~\\
        ~ & ~ & ~
    \end{ytableau}} && 10 \nonumber \\
    &&&\lambda - 6 s_1 =0  &&&&{\tiny \begin{ytableau}
        ~ & ~ & ~& ~& ~& ~
    \end{ytableau}} && 28 \nonumber \\
    &&&\lambda^2 - 2 s_1 \lambda + 3s_2 = 0  &&&&{\tiny \begin{ytableau}
        ~ & ~ \\
        ~ 
    \end{ytableau}} && 8 \nonumber \\
    &&&\lambda^2 - 8 s_1 \lambda + 12\left(s_1^2 + s_2\right) = 0  &&&&{\tiny \begin{ytableau}
        ~ & ~& ~ &~ & ~ \\
        ~ 
    \end{ytableau}} && 35 \nonumber \\
    &&&\lambda^3 - 8 s_1 \lambda^2 + \left(12 s_1^2 + 28 s_2\right)\lambda - 48 s_1s_2 - 80 s_3 = 0 &&&&{\tiny \begin{ytableau}
        ~ & ~& ~ &~  \\
        ~ & ~
    \end{ytableau}} && 27 \nonumber
\end{align*}
Using these explicit expressions, one can check that the above remarks about properties of the spectrum  hold true.

\subsection{
Bethe ansatz and Heun polynomials.} 
In this section, we will calculate the spectrum of the Laplace-Beltrami operator for an arbitrary left-invariant metric~(\ref{F3metric}) on $\mathcal{F}(3)$. To this end, we first need to recall that, according to Proposition~\ref{spinFlagConnection}, $\lim_{p\to \infty}\mathrm{Spec}(p)$ is precisely the spectrum of the Laplace-Beltrami operator. Thus, we only need to study the spectrum of the $\mathrm{SU}(3)$ spin chain, which can be computed via the Bethe ansatz. 

Let us now recall the basic formulas of the Bethe ansatz method\footnote{We refer to \cite{Feigin1994-ag,Mukhin2007-rl} for a general introduction and details and to \cite{AMBURG2025105436} for geometric perspective on $\mathfrak{sl}(2,\mathbb{C})$ Bethe ansatz.} in the case of the $\mathrm{SU}(3)$ spin chain we are dealing with. First, we fix the notation: $\omega_1,\omega_2$ are the two fundamental weights of $\mathrm{SU}(3)$ and $\alpha_1 =2 \omega_1 - \omega_2,\,\alpha_2 = -\omega_1 + 2\omega_2$ are the simple roots \cite{fulton1991representation}. The highest weight of $\mathrm{Sym}(p)$ is $p\,\omega_1$. Recall that, in Bethe ansatz, this is also the ferromagnetic `ground state' (often referred to as the pseudo vacuum) corresponding to zero Bethe roots. Each irreducible representation $\mathcal{R} \in \mathrm{Irr}(\mathrm{Hilb}(p))$ has a highest weight of the form $3p\,\omega_1 - \ell\alpha_1\,- \ell'\alpha_2 = n\,\omega_1 + m\,\omega_2$, where $\ell, \ell',n,m \in \mathbb{N}_{0}$. Here $\ell$ and $\ell'$ are the numbers of Bethe roots corresponding to this representation, whereas $n$ and $m$ are related to the Young diagram of $\mathcal{R}$, representing the number of one-box and two-box columns in the diagram, respectively (see~(\ref{nmdiagram})). These are related as follows: 
\bea
\ell = {1\over 3}\left(6p-(2n+m)\right)\,, \quad\quad \ell' = {1\over 3}\left(3p-(n+2m)\right)\,.
\eea
In particular, according to~(\ref{nmpstates}) the states in $\mathrm{Hilb}(p+1)\backslash \mathrm{Hilb}(p)$ correspond to $\ell'=0$ (and $\ell=m$).

In order to calculate the eigenvalues of the Gaudin Hamiltonians (\ref{GaudinHam1})-(\ref{GaudinHam3}), one should first solve the Bethe equations
\begin{align}
    & p\left(\frac{1}{t_a - \alpha_{23}} + \frac{1}{t_a - \alpha_{13}}+\frac{1}{t_a - \alpha_{12}}\right) - 2\sum_{\substack{b=1\\b\neq a}}^{\ell}\frac{1}{t_a - t_b} + \sum^{\ell'}_{\alpha=1}\frac{1}{t_a - k_\alpha} = 0\,,\label{BetheEq}\\
    &- 2\sum_{\substack{\beta=1\\\beta\neq \alpha}}^{\ell'}\frac{1}{k_\alpha - k_\beta} + \sum^{\ell}_{b=1}\frac{1}{k_\alpha-t_b} = 0\,,\nonumber\\
    &\text{for} \quad a = 1,2,\dots,\ell\,;\quad \alpha =1,2,\dots,\ell' \,,\nonumber
\end{align}
where $t_a, k_\alpha$ are the  Bethe roots. Clearly, they are pairwise distinct (or else the equations do not make sense). Note that when $\ell' =0$ one arrives at the Bethe equations of the $\mathrm{SU}(2)$ spin chain. Then the eigenvalues $\mu_i$ of the  Gaudin Hamiltonians $\mathrm{H}_i$ (up to constant terms) are expressed as
\begin{align}
    \mu_i = p\sum_{a=1}^{\ell}\frac{1}{t_a-z_i}\,\;\;\text{for}\;\;i=1,2,3\,,
\end{align}
where $z_1 = \alpha_{23}, z_2 = \alpha_{13}$ and $z_3 = \alpha_{12}$ are the Gaudin parameters. One can check that $\mu_1+\mu_2+\mu_3=0$.

The Bethe equations (\ref{BetheEq}) are quite complex and provide a lot more information than is needed to determine the spectrum. Thus, we will use an equivalent reformulation \cite{Mukhin2007-rl} (see also \cite{BABUJIAN_1994}) based on a classic result of Stieltjes \cite{Stieltjes_1885} (see also \cite[Section $6.8$]{szego1939orthogonal}). Let us recall the reformulation in the context of our spin chain. For given Bethe roots $t_a$'s and $k_\alpha$'s, one defines the polynomials 
\bea
\mathcal{P}_1(x) := \prod_{a=1}^{\ell}(t_a\,-\,x)\quad \textrm{and}\quad \mathcal{P}_2(x) := \prod_{\alpha=1}^{\ell'}(k_\alpha\,-\,x)\,.
\eea
As a consequence of the Bethe equations~(\ref{BetheEq}), these  satisfy in turn the following differential equation\footnote{Compared to the original article \cite{Mukhin2007-rl}, we have shifted the $\mu_i$'s.} 
\begin{align}
    \mathcal{P}_1''\mathcal{P}_2 - \mathcal{P}_1'\mathcal{P}_2' + \mathcal{P}_1\mathcal{P}_2'' - p\,&\mathcal{P}_1'\mathcal{P}_2\left(\frac{1}{x-\alpha_{23}}+\frac{1}{x-\alpha_{13}}+\frac{1}{x-\alpha_{12}}\right)+\nonumber\\
    -&\mathcal{P}_1\mathcal{P}_2\left(\frac{\mu_1}{x-\alpha_{23}}+\frac{\mu_2}{x-\alpha_{13}}+\frac{\mu_3}{x-\alpha_{12}}\right)=0\,,\label{BetheForSU3}
\end{align}
Conversely, if one can find the $\mu_i$'s such that $\mu_1+\mu_2+\mu_3=0$ and the solutions $\mathcal{P}_1$ and $\mathcal{P}_2$ are polynomials with distinct roots, then the $\mu_i$'s are the eigenvalues of the Gaudin Hamiltonians. 

Let us now analyze (\ref{BetheForSU3}). As  noted above, for given $p$, we are interested in the spectrum for the representations in $\mathrm{Hilb}(p)\backslash\mathrm{Hilb}(p-1)$. For these states one has $\ell' = 0$ so that $\mathcal{P}_2\equiv 1$, and we arrive at the simpler equation (here $\mathcal{P}\equiv \mathcal{P}_1$)
\begin{align}
    \mathcal{P}'' - p\,&\mathcal{P}'\left(\frac{1}{x-\alpha_{23}}+\frac{1}{x-\alpha_{13}}+\frac{1}{x-\alpha_{12}}\right)+\nonumber\\
    -&\mathcal{P}\;\left(\frac{\mu_1}{x-\alpha_{23}}+\frac{\mu_2}{x-\alpha_{13}}+\frac{\mu_3}{x-\alpha_{12}}\right)=0\,,\label{HeunEq}
\end{align}
which is the well-known Heun equation -- a Fuchsian equation with four regular singularities on the Riemann sphere (see \cite[Chapter 31]{NIST} and \cite[Section $15.3$]{bateman1955higher} for general information on the Heun equation). Our goal is to find its polynomial solutions of degree $\ell$ with distinct roots. To this end it will be more convenient to work with the standard form of the Heun equation. We thus make the  substitution $x\mapsto \alpha_{23}-(\alpha_{23}-\alpha_{12})\,x$, so that the equation takes the form
\begin{align}
    &\mathcal{P}''-p\,\mathcal{P}'\left(\frac{1}{x}+\frac{1}{x-1}+\frac{1}{x-a}\right)+\mathcal{P}\frac{\alpha\beta\, x- q}{x(x-1)(x-a)} =0\,,\label{HeunFinalForm}\\
    &\text{where} \quad \alpha+\beta+1= -3p\,,\nonumber\\
    & q = \mu_1(\alpha_{13}-\alpha_{23})\,,\;\alpha\beta =\mu_3(\alpha_{23}-\alpha_{12})+\mu_2(\alpha_{23}-\alpha_{13})\,,\nonumber\\
    &a = \frac{\alpha_{23}-\alpha_{13}}{\alpha_{23}-\alpha_{12}}.\nonumber
\end{align}
In general, (\ref{HeunFinalForm}) has a polynomial solution of degree $\ell$ if and only if 
\bea
\alpha = -\ell
\eea
and $q$ is an eigenvalue of the  tridiagonal matrix (see Appendix \ref{polSolHeunEq} for details)
\begin{align}
    \mathcal{M}_{\ell}:=&\begin{pmatrix}
        0 & R_0 & 0 & \dots & 0\\
        P_1 &  -Q_1 & R_1 & \dots  & 0 \\\
        0 & P_2 & -Q_2 & & \vdots\\
        \vdots & \vdots & &\ddots & R_{\ell-1} \\
        0 & 0 & \dots & P_{\ell} & -Q_{\ell}
    \end{pmatrix}\,,\\
    \nonumber\\
   \text{where} \quad &P_j = (j-1+\alpha)(j-1+\beta)\,,\\
  &Q_j = j(j-2p-1)(1+a)\,,\\
   &R_j = a(j+1)(j-p)\,.
\end{align}
The eigenvector of $\mathcal{M}_{\ell}$ is then the vector of coefficients of the polynomial solution. 

There are $\ell+1$ independent polynomial solutions, each corresponding to one eigenvalue/eigenvector of $\mathcal{M}_{\ell}$. We are interested in solutions with distinct roots. In general they are hard to filter out but we note that the multiplicity for irreducible representations with\footnote{From Section \ref{compSpec}, we know that it is enough to study only the representations with $n\geq m$.} $n\geq m$ is $m+1 = \ell+1$. This means that every eigenvector of $\mathcal{M}_\ell$ must correspond to one of these representations, and therefore all corresponding polynomials must be admissible, i.e. with distinct roots. 

Now we are ready to describe the eigenvalues $\varepsilon$ of $\mathcal{H}$ corresponding to the representations with $\ell \in \mathbb{N}_0,\,\ell' = 0$ and $n \geq m$ for $p\in \mathbb{N}_0$. From the above, $\alpha = -\ell$ and $\beta = \ell-3p-1$. Thus, using $\mu_1+\mu_2+\mu_3 = 0$ and  (\ref{HeunFinalForm}), one can express the  $\mu_i$'s via~$q$. Then, upon substituting into (\ref{fullHam}), one arrives at  \begin{align}
    &\varepsilon=q(\alpha_{23}-\alpha_{12})+s_1 p(p+1)+\alpha_{23}(\ell^2-\ell(3p+1))\,,\nonumber\\
    &\mathrm{det}\left(\mathcal{M}_{\ell}-q\,\mathrm{Id}_{\ell+1}\right) = 0\,,\label{euationsOnEigenvalues}
\end{align}
where $\mathrm{Id}_{\ell+1}$ is an $(\ell+1)\times (\ell+1)$ identity matrix. We have also made use of the explicit expression for the eigenvalue $\mathcal{C}_2(3p-\ell,\ell,0)$ of the quadratic Casimir operator $\mathcal{C}_2$ on a given representation, which can be found using the Perelomov-Popov formula \cite{PerelomovPopov}
\begin{gather}\label{PPformula}
    \mathcal{C}_2(p_1,p_2,p_3)=\sum_{j=1}^3 v_j(v_j - 2j),\quad v_i = p_i - \frac{1}{3} \sum_{j=1}^3 p_j.
    \end{gather} 
We have chosen the additive constant in $\mathcal{H}$ so that the eigenvalue $\varepsilon$ for the $(3p,0)$ representation would coincide with the one obtained in~(\ref{3prepeigenvalue}).

The above polynomial equations may be solved for the eigenvalues $\varepsilon$ of the Hamiltonian $\mathcal{H}$ corresponding to fixed $p$ and $\ell\leq p$. 
Therefore, to obtain the full spectrum of the Laplace-Beltrami operator for the metric~(\ref{F3metric}) one should vary $p=0, 1, 2, \ldots$ and $\ell=0, 1, \ldots , p$. 

Let us summarize the discussion above into one proposition:
\begin{prop}
    Eigenvalues $\varepsilon$ of $-\triangle$ are given by  (\ref{euationsOnEigenvalues}) for $p \in \mathbb{N}_0$ and ${\ell \in \{0,1,\dots,p\}}$, each with  multiplicity\footnote{For $\ell\neq p$, this is the doubled dimension of the representation $(3p-2\ell,\ell)$  since  the eigenvalues on $(n,m)$ and $(m,n)$ representations are the same.} $(3p-2\ell+1)(\ell+1)(3p-\ell+2)$ for $\ell \neq p$ and $(p+1)^3$ for $\ell=p$. 
\end{prop}

\section{Spectral reconstruction} \label{spectrRecon}

 We will now describe another analytic method for extracting information about properties of the spectrum. Although this method cannot provide the complete spectrum, it can accurately determine some properties. In particular, we will use it to test the general solution obtained by means of Bethe ansatz in the previous section. 
 
 The key idea is to utilize the $\mathrm{S}_3$-invariance and homogeneity of the spin system under consideration, as well as the known spectrum for a particular choice of $\alpha_{ij}$. We will refer to the method as spectral reconstruction.

Let us take a closer look at the spin system. One can define the action of $\mathrm{S}_3$ in our system via permutation of the sites, since there is the same representation at each site. The symmetry acts on the spin operators via $\mathrm{S}^a_i \mapsto \mathrm{S}^a_{\sigma(i)}$ where $\sigma \in \mathrm{S}_3$. Effectively, it works as a  permutation of the $\alpha_{ij}$'s:
\begin{align}
    \mathcal{H} = \sum_{i<j}\alpha_{ij}\mathrm{S}^a_i\mathrm{S}^a_j\;\xrightarrow{\sigma\, \in\, \mathrm{S}_3} \; \sum_{i<j}\alpha_{\sigma^{-1}(i)\sigma^{-1}(j)}\mathrm{S}^a_i\mathrm{S}^a_j\,,
\end{align}
where we assume that $\alpha_{ij} = \alpha_{ji}$. Let us relabel the $\alpha_{ij}$'s as  $\alpha_{12}=\alpha^{(3)}$, $\alpha_{13}=\alpha^{(2)}$ and $\alpha_{23}=\alpha^{(1)}$. In these terms, the map $\alpha_{ij} \mapsto \alpha_{\sigma^{-1}(i)\sigma^{-1}(j)}$ reads as $\alpha^{(i)} \mapsto \alpha^{(\sigma^{-1}(i))}$. 
It follows that the spectrum of $\mathcal{H}$ should be invariant under the symmetry. Thus, eigenvalues are symmetric ($\mathrm{S}_3$-invariant) functions of $\alpha_{ij}$'s, which may therefore be expressed in terms of the elementary symmetric polynomials $s_1, s_2, s_3$ introduced in~(\ref{s1s2s3}).  

Next, one notes that under the map $\alpha_{ij}\mapsto \zeta\,\alpha_{ij}$ an eigenvalue $\lambda$ of $\mathcal{H}$ is scaled as~$\zeta\,\lambda$. Therefore, eigenvalues are homogeneous functions of $\alpha_{ij}$'s with degree~$1$. Finally, the operators $\mathrm{S}^a_i\mathrm{S}^a_j$ are simply matrices, which means that the characteristic polynomial of~$\mathcal{H}$ contains $\alpha_{ij}$'s in a polynomial way. 

To summarize the above, for an irreducible representation $\mathcal{R}\in \mathrm{Irr}\left(\mathrm{Hilb}(p)\right)$ with multiplicity $n$, the characteristic polynomial of the restricted Hamiltonian may be written as 
\begin{align}\label{genericChar}
    \lambda^n+\sum_{k=1}^{n} \lambda^{n-k} \sum_{b_1 + 2 b_2 + 3 b_3 = k} C^k_{b_1\,b_2\,b_3}\;s_1^{b_1}s_2^{b_2}s_3^{b_3} = 0\,,
\end{align}
where we have taken into account both homogeneity and $\mathrm{S}_3$-invariance. Each root of the polynomial appears exactly $\mathrm{dim}\left(\mathcal{R}\right)$ times in $\mathrm{Spec}(p)$. The problem of finding the spectrum is thus reduced to determining the coefficients $C^k_{b_1\,b_2\,b_3}$.

In order to find $C^k_{b_1\,b_2\,b_3}$, we will use the additional information coming from the fact that the spectrum is known when two of $\alpha_{ij}$'s are equal. Suppose  $\alpha_{13} = \alpha_{23}$. Then, up to an additive constant, $\mathcal{H}$ is a sum of two commuting Casimir operators, i.e. 
\begin{align}
    \mathcal{H} = \frac{\alpha_{12}-\alpha_{13}}{2}\left(\mathrm{S}_1^a+\mathrm{S}_2^a\right)^2 + \frac{\alpha_{13}}{2}\left(\mathrm{S}_1^a+\mathrm{S}_2^a+\mathrm{S}_3^a\right)^2\,.\label{specialHam}
\end{align}
There is the well-known formula~(\ref{PPformula}) for the eigenvalues of Casimir operators. Using it, one arrives at the following proposition:
\begin{prop}[\cite{Bykov_2024}]
    For a fixed value of $p$, the eigenvalues of $\mathcal{H}$ (\ref{specialHam}) on the irreducible representation $\mathcal{R}\in \mathrm{Irr}\left(\mathrm{Hilb}(p)\right)$, with Young diagram having row lengths $\left(p_1^B, p_2^B, p_3^B\right)$, take the form
    \begin{gather}\label{F3specSp}
        \Lambda_k = \frac{\alpha_{13}}{2}\mathcal{C}_2(p^B_1,p^B_2,p^B_3)+
        \frac{\alpha_{12} - \alpha_{13}}{2}
        \left[
            \mathcal{C}_2\left(p_1^A, p^A_2, 0\right)-\mathcal{C}_2(p,p,0)
        \right], 
    \end{gather}
    where $p_1^A+p_2^A=2p,\, p_2^A\leq p\leq p_1^A$, $p_1^B+p_2^B+p_3^B=3p,\, p_2^B\leq p_1^A\leq p_1^B,\, p_3^B\leq p_2^A \leq p_2^B$. Each $\Lambda_k$ has a multiplicity of $\mathrm{dim}(\mathcal{R})$. 
\end{prop}

For simplicity, we will assume $\alpha_{12}=1,\, \alpha_{13} = \alpha_{23} = q$. Then, for a particular irreducible representation $\mathcal{R}\in \mathrm{Irr}\left(\mathrm{Hilb}(p)\right)$ with multiplicity $n$ one forms its characteristic polynomial using (\ref{F3specSp})
\begin{align}
    \mathfrak{L}\left(\mathcal{R}|q\right)=\prod_{k=1}^{n} \left(\lambda - \Lambda_k(q)\right)\,.\label{specRecon}
\end{align}
From (\ref{F3specSp}) we know that $\Lambda_k~=~\chi_k\,+\,\nu_k\,q$. Thus, one finds the following decomposition for $\mathfrak{L}\left(\mathcal{R}|q\right)$:
\begin{align}
    &\mathfrak{L}\left(\mathcal{R}|q\right) = \lambda^n+\sum_{k=1}^n\left(-1\right)^k f_k(q)\, \lambda^{n-k}\,,\label{descrFromSpCase}\\
    &\text{where}\;\;f_k(q) := \sum_{i_1<i_2<\dots<i_k}\;\prod_{j=1}^k\Lambda_{i_j}(q)=\sum_{j=0}^{k}p_j\,q^j\,.\nonumber
\end{align}
In principle, the $p_j$'s are known. From (\ref{genericChar}) in the case when  $\alpha_{12} = 1$ and $\alpha_{13}=\alpha_{23}=q$ one finds that the same polynomial admits a different description:
\begin{align}
    &\mathfrak{L}\left(\mathcal{R}|q\right) = \lambda^n+\sum_{k=1}^{n} \lambda^{n-k}  g_k(q)\,,\label{descrGeneral}\\
    &\text{where}\;\;g_k(q):=\sum_{b_1 + 2 b_2 + 3 b_3 = k} C^k_{b_1\,b_2\,b_3}\;\left(1+2q\right)^{b_1}\left(2q+q^2\right)^{b_2}\left(q^2\right)^{b_3}\,.\nonumber
\end{align}
Comparing (\ref{descrFromSpCase}) and (\ref{descrGeneral}), one obtains equations on the polynomials: $g_k(q) = (-1)^k f_k(q)$, where $k=1, \dots, n$. The coefficients  must be equal, which provides linear constraints on the $C^k_{b_1 b_2 b_3}$'s, since  the coefficients of $f_k(q)$ are known. 

Next we would like to understand when one can completely determine the $C^k_{b_1 b_2 b_3}$'s. We start by noting that the number of constraints provided by $f_k(q)$ is $k+1$. At the same time, the number of $C^k_{b_1 b_2 b_3}$'s depends non-linearly on $k$. One shows that it is less than $k+1$ only for\footnote{At $k=7$ the number of $C^k_{b_1b_2b_3}$ is $8$, whereas for $k=8$ it is $10$.} $k<8$. Thus, for $k \geq 8$ it is not possible to determine all of the $C^k_{b_1b_2b_3}$'s. For $k = 6$ and $k = 7$, the number of $C^k_{b_1b_2b_3}$'s is exactly $k + 1$, but the system of linear equations is degenerate. As a result, it is as well impossible to find all the $C^k_{b_1b_2b_3}$'s. For $k \leq 5$, the number of $C^k_{b_1b_2b_3}$'s is $k$, so in this case one has a rectangular system of linear equations on the $C^k_{b_1b_2b_3}$'s. One easily finds that the rank of this system is maximal, so that one can find $C^k_{b_1b_2b_3}$ explicitly\footnote{The  fact that the system of equations is rectangular means that there is a constraint on the $p_j$'s. This constraint should be interpreted as a condition on the values of the Casimir operator (see (\ref{F3specSp})).}.

To summarize the above, we formulate the following proposition: 

\begin{prop}
    For $k \leq 5$, $C^k_{b_1 b_2 b_3}$'s can be reconstructed from the known spectrum of the spin chain in the case  $\alpha_{13} = \alpha_{23}$. 
\end{prop}

As a result, the characteristic polynomial can be fully reconstructed for those irreducible representations $\mathcal{R}\in \mathrm{Irr}\left(\mathrm{Hilb}(p)\right)$ which have a multiplicity of at most $5$. Consequently, for $p\leq 4$, the spectrum of the spin chain can be completely reconstructed (we write out the resulting characteristic polynomials in Appendix \ref{p4specRecon}).

\begin{figure}[H]
    \centering
    \includegraphics[width=0.5\textwidth]{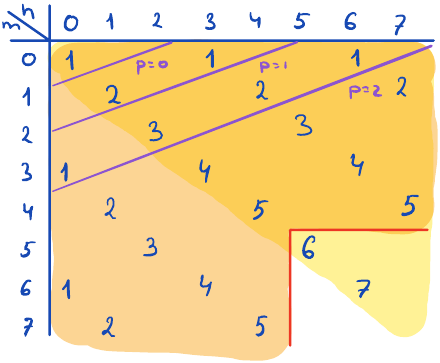}
    \caption{Representations for which the spectrum can be explicitly determined using spectral reconstruction (the hook between the blue axes and the red lines).}
    \label{table2}
\end{figure}

Let us consider a simple example of irreducible representation entering with multiplicity $1$. From Table \ref{table1}, we learn that these are representations of type  $(3k,0)$ and $(0,3k)$ for $k\in \mathbb{N}_{0}$. One shows that $f_1(q)= (1+2q)(k^2+k)$ in both cases and, thus, the characteristic polynomial for these representations is
\begin{align}
    \lambda -(k^2+k)s_1=0 \,,
\end{align}
which reproduces the result~(\ref{3prepeigenvalue}).

As a simple test, one checks that the polynomials obtained using spectral reconstruction for $p\leq 4$ are the same as those obtained using the Bethe ansatz.

\section{Conclusion and outlook}

In this paper, we have explicitly described the geodesics and the spectrum of the Laplace-Beltrami operator on $\mathcal{F}(3)$, equipped with an arbitrary invariant metric, by using the correspondence between one-dimensional sigma models on flag manifolds and $\mathrm{SU}(N)$ spin chains. In Section \ref{SUNSpinChain}, we pointed out that our methods are generally applicable to all full flag manifolds equipped with a metric from a specific class (i.e. such that $\alpha_{ij} = \frac{t_i - t_j}{z_i-z_j}$ in (\ref{FNmetric})). 

In fact, one can as well apply our methods to the case of partial flag manifolds. Consider the partial flag manifold $\mathcal{F}_{n_1,n_2\dots,n_k}$, where $n_1+\dots+n_k=N$. In this case there is a correspondence with an $\mathrm{SU}(N)$ spin chain (see Propositions \ref{subbundle} and \ref{spinFlagConnection}), where, in terms of Young diagrams, the Hilbert space  is~\cite{Bykov_2024}
\begin{align}
    n_1 \quad
        \underbrace{\!\!\!\!\!
        \begin{cases}
            \begin{ytableau}
            ~ & \dots & ~ \\
            \vdots & \vdots & \vdots\\
            ~ &  \dots & ~ \\
            \end{ytableau}
        \end{cases} \!\!\!\!\!\!}_{p} \quad \otimes\quad n_2 \quad
        \underbrace{\!\!\!\!\!
        \begin{cases}
            \begin{ytableau}
            ~ & \dots & ~ \\
            \vdots & \vdots & \vdots\\
            ~ &  \dots & ~ \\
            \end{ytableau}
        \end{cases} \!\!\!\!\!\!}_{p}\quad \otimes \quad\dots \quad \otimes \quad n_k \quad
        \underbrace{\!\!\!\!\!
        \begin{cases}
            \begin{ytableau}
            ~ & \dots & ~ \\
            \vdots & \vdots & \vdots\\
            ~ &  \dots & ~ \\
            \end{ytableau}
        \end{cases} \!\!\!\!\!\!}_{p}\,.\label{HilbPartialFlag}
\end{align}
The Hamiltonian of the system is defined in~(\ref{spinHamN}). As we noted in Section \ref{SUNSpinChain}, such spin system with three sites is Gaudin and, therefore, integrable. We thus conclude that the one-dimensional sigma model on $\mathcal{F}_{n_1,n_2,n_3}$ is also integrable, both at the classical and quantum level. Let us point out that, from a general perspective, the case of partial flag manifolds can be treated as a special case of full flag manifolds, where some of the $\alpha_{ij}$'s tend to infinity (see \cite{Bykov_2024} for details).

Another generalization is related to magnetic geodesics and the magnetic Laplacian (see \cite{Bolsinov2006,Kordyukov_2019} for the background). To this end one should consider spin chains with varying widths of Young diagrams  at different sites (see (\ref{Hilbp}) and (\ref{HilbPartialFlag})). For example, in the case of $\mathcal{F}(N)$ we are interested in 
$\otimes_{i=1}^N \mathrm{Sym}(p_i)$, where $p_i = p+q_i$. In the limit $p\to \infty$ the spin chain reproduces a one-dimensional sigma model on a flag manifold with a magnetic field, where $q_i$'s are the magnetic charges \cite{IsotropicBykov}. We expect that the methods described in the present paper are also valid in the magnetic case.


\vspace{0.5cm}
\textbf{Acknowledgments.} Sections 1-4 were supported by the Russian Science Foundation grant № 25-72-10177 (\href{https://rscf.ru/en/project/25-72-10177/}{\emph{https://rscf.ru/en/project/25-72-10177/}}). Sections 5-7 were supported by the Moscow Center of Fundamental and Applied Mathematics of Lomonosov Moscow State University under agreement № 075-15-2025-345. We would like to thank  I.~Taimanov and I.~Tolstukhin for valuable discussions.

\vspace{2cm}
\appendixbig

\setcounter{section}{10}
\newcounter{appcounter}
\setcounter{appcounter}{1}
\renewcommand{\thesection}{\Alph{appcounter}}

\section{Weierstrass $\zeta, \sigma$ and $\wp$ elliptic functions}

Here, following \cite{Akhiezer_1990} and \cite{Pastras_2020}, we recall the definitions of the elliptic functions that we use throughout the paper.

Let us fix two complex numbers $\omega_1,\omega_2 \in \mathbb{C}$ such that\footnote{In our applications, we will assume that $\omega_1 \in \mathbb{R}$ and $\omega_2 \in i\mathbb{R}$.} $\mathrm{Im}\left(\omega_2/\omega_1\right) > 0 $. Then, one defines the Weierstrass $\wp$-function using the series
\begin{align}
    \wp(z) = \frac{1}{z^2} + \sum_{\{m,n\}\neq\{0,0\}}\left(\frac{1}{\left(z+2m\omega_1+2n\omega_2\right)^2} - \frac{1}{\left(2m\omega_1+2n\omega_2\right)^2}\right)\,.
\end{align}
This is a doubly periodic function $\wp(z+2\omega_1)=\wp(z+2\omega_2)=\wp(z)$ with second-order poles at $2m\omega_1+2n\omega_2$, where $m,n \in \mathbb{Z}$. The Weierstrass $\wp$-function satisfies the equation
\begin{align}
   &\left(\frac{\mathrm{d}\;\,}{\mathrm{d}z}\wp\right)^2 = 4 \wp^3-g_2\,\wp - g_3\,,\\
   &\text{where}\;\; g_2 = 60 \sum_{\{m,n\}\neq\{0,0\}} \frac{1}{\left(2m\omega_1+2n\omega_2\right)^4}\,,\nonumber\\
   &g_3 = 140 \sum_{\{m,n\}\neq\{0,0\}} \frac{1}{\left(2m\omega_1+2n\omega_2\right)^6}\,.\nonumber
\end{align}
Now, we can define the Weierstrass $\zeta$- and $\sigma$-functions as follows:
\begin{align}
    &\zeta(z) = \frac{1}{z} - \int_{0}^z\left(\wp(u)-\frac{1}{u^2}\right)\mathrm{d}u\,,\\
    &\ln\left(\frac{\sigma(z)}{z}\right) = \int^{z}_{0}\left(\zeta(u)-\frac{1}{u}\right)\mathrm{d}u\,,
\end{align}
where the integration contour should not pass through $2m\omega_1+2n\omega_2$ except for the case $m=n=0$. These are quasi-periodic functions:
\begin{align}
    &\zeta(z+2\omega_i) = \zeta(z)+2\zeta(\omega_i)\,,\\
    &\sigma(z+2\omega_i) = -e^{2\zeta(\omega_i)(z+\omega_i)}\sigma(z)\,,\;\text{for}\;i=1,2\,.
\end{align}

\stepcounter{appcounter}
\section{On the reality of roots}\label{realRoots}
Here we will show that all  roots of the polynomial $4\mathrm{Y}^3-g_2\mathrm{Y}-g_3$ are real, if $g_2$ and $g_3$ are defined as in (\ref{WeierstrassY}). This step is crucial for understanding the properties of the Weierstrass $\wp$-function, which is a solution to (\ref {WeierstrassY}).

First we return to the variable $y = \mathrm{Y} + \frac{1}{3}\left(\mathrm{I}_2+\mathrm{I}_3\right)$, so that we are led to the following system of equations (see (\ref{origEqOny})):
\begin{align}
    &4y\left(y-\mathrm{I}_2\right)\left(y-\mathrm{I}_3\right)=-\mathrm{I}_4\,,\label{yeqApp}\\
    &\text{where}\quad \mathrm{I}_2 =|c|^2+\frac{\alpha_{12}-\alpha_{13}}{\alpha_{23} - \alpha_{13}}\,|a|^2\,,\nonumber\\
    &\mathrm{I}_3 =|c|^2+\frac{\alpha_{12}-\alpha_{13}}{\alpha_{12} - \alpha_{23}}\,|b|^2\,,\nonumber\\-&\mathrm{I}_4 = \frac{\left(\alpha_{12} - \alpha_{13}\right)^2}{\left(\alpha_{12} - \alpha_{23}\right)\left(\alpha_{23} - \alpha_{13}\right)}\,\left(\Bar{a}b\Bar{c} + a\Bar{b}c\right)^2\,.\nonumber
\end{align}
It is convenient to forget for a moment about the meaning of $a,b,c$ and treat them as arbitrary complex numbers. We note that $\mathrm{I}_2 \geq |c|^2$, $\mathrm{I}_3 \geq |c|^2$ and $-\mathrm{I}_4 \geq 0$ because $\alpha_{12} > \alpha_{23} > \alpha_{13}$. Note that if the equation 
\begin{align}
    &4y\left(y-\mathrm{I}_2\right)\left(y-\mathrm{I}_3\right)=\mathrm{J}\,,\label{Deq}\\
    &\text{where}\quad \mathrm{J} \geq -\mathrm{I}_4 \geq 0\,,
\end{align}
has only real roots, then the original equation (\ref{yeqApp}) will also have only real roots (see~Fig.~\ref{roots}). 

\begin{figure}[h]
    \centering
    \includegraphics[width=0.35\textwidth]{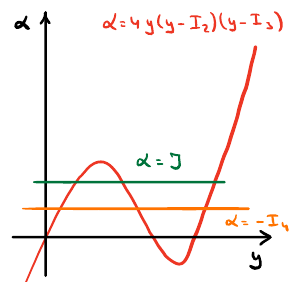}
    \caption{Roots of the polynomials.}
    \label{roots}
\end{figure}

We choose 
\begin{align}
    \mathrm{J} = \frac{\left(\alpha_{12} - \alpha_{13}\right)^2}{\left(\alpha_{12} - \alpha_{23}\right)\left(\alpha_{23} - \alpha_{13}\right)}\,4|a|^2|b|^2|c|^2\,,
\end{align}
It is easy to see that $y=|c|^2$ solves (\ref{Deq}), but $|c|^2 \leq \mathrm{I}_2$ and $|c|^2 \leq \mathrm{I}_3$, therefore equation~(\ref{Deq}) has three real roots. We also note that all three roots are non-negative.

Let us now study the degenerate case when two\footnote{One easily sees that the case of three equal roots corresponds to $|a|=|b|=|c|=0$.} of the roots of equation (\ref{yeqApp}) are equal. This happens when $\mathrm{J}=-\mathrm{I}_4$, and the maximum of the function $4y(y-\mathrm{I}_2)(y-\mathrm{I}_3)$ is reached at $y=|c|^2$. We illustrate this situation using the following figure:
\begin{figure}[H]
    \centering
    \includegraphics[width=0.37\textwidth]{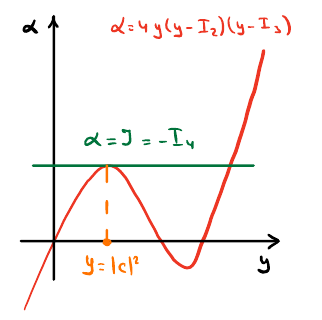}
    \caption{Degenerate case.}
    \label{degCase}
\end{figure}

One easily expresses the above conditions in the following algebraic way:
\begin{align}
    &\left(\alpha_{12}-\alpha_{13}\right)|a|^2|b|^2+\left(\alpha_{23}-\alpha_{12}\right)|a|^2|c|^2 + \left(\alpha_{13}-\alpha_{23}\right)|b|^2|c|^2=0\,,\\
    &\Bar{a}b\Bar{c}=a\Bar{b}c\,.\nonumber
\end{align}

\stepcounter{appcounter}
\section{On the boundedness of solutions}\label{boundSolution}
We shall verify that the solutions (\ref{ABCsolutions}) are bounded, as they should be in order to describe geodesic flow. Note that this is not obvious at first sight,  since, as we noted in Section~\ref{geodSection}, $|a|^2,|b|^2$ and $|c|^2$ are  expressed through the Weierstrass $\wp$-function, which has poles. Concretely, we shall check that $y_{\pm}(\tau|\mu)$ are bounded as functions of $\tau$.

Due to the quasiperiodicity of $\sigma$, i.e. $\sigma(z+2\,\omega_1) = - e^{2\zeta(\omega_1)(z+\omega_1)}\sigma(z)$, the function $y_{\pm}(\tau|\mu)$ is also quasiperiodic:
\begin{align}
    y_{\pm}(\tau+2\,\omega_1|\mu) = e^{\mp 2 \left(\omega_1 \zeta(\mu) - \mu \zeta(\omega_1)\right)}\,y_{\pm}(\tau|\mu)\,.
\end{align}
Thus, boundedness depends on the properties of $\phi(\mu):=\omega_1\zeta(\mu) - \mu \zeta(\omega_1)$: 
$y_{\pm}(\tau|\mu)$ is bounded\footnote{In fact, the solution can be interpreted as a Bloch wave \cite{grosso2013solid} for a periodic potential generated by the  $\wp$-function.} if and only if $\phi(\mu)$ is pure imaginary. It turns out that $\phi(\mu)$ is imaginary when $\lambda = \wp(\mu) \in (-\infty,e_3]\cup[e_2,e_1]$, where $e_1 > e_2 > e_3$ are the roots of ${4\mathrm{Y}^3 - g_2 \mathrm{Y}-g_3}$~\cite{Pastras_2020}. In our case $\lambda_1 = 2\mathrm{I}_2/3-\mathrm{I}_3/3$, $\lambda_2 = -\mathrm{I}_2/3+2\mathrm{I}_3/3$ and $\lambda_3 = -\mathrm{I}_2/3-\mathrm{I}_3/3$. 

It is simpler to work with the roots $e_i'$'s of the polynomial in $y = \mathrm{Y} + (\mathrm{I}_2 + \mathrm{I}_3) / 3$, which is given by $4y(y - \mathrm{I}_2)(y - \mathrm{I}_3) + \mathrm{I}_4$ (see (\ref{yeqApp})). We shall check that $\lambda'_i=\lambda_i+(\mathrm{I}_2 + \mathrm{I}_3) / 3 \in (-\infty,e'_3]\cup[e'_2,e'_1]$. Clearly, this is true because $\lambda'_1 = \mathrm{I}_2, \lambda'_2=\mathrm{I}_3$ and $\lambda'_3 = 0$ and, from Fig.~\ref{roots}, $e'_3 \geq 0$ and $[\mathrm{I}_2,\mathrm{I}_3]\subset[e'_2,e'_1] $.

\stepcounter{appcounter}
\section{Spectrum for $p = 4$ from spectral reconstruction}\label{p4specRecon}
In this Appendix we write out the characteristic polynomials for the irreducible representations in $\mathrm{Hilb}(4)$ obtained via spectral reconstruction. We will restrict to polynomials for the representations $(n, m)$ with $n \geq m$, because the polynomials for $(n, m)$ and $(m, n)$ are the same: 
\begin{align*}
    &\textrm{Mult.} &&&\textrm{Rep.}  &&&& \textrm{Characteristic polynomial}\\
    \hline
    &1 &&& (0,0) &&&& \lambda= 0 \\
    & &&& (3,0) &&&& \lambda - 2s_1=0 \\
    & &&& (6,0) &&&& \lambda - 6s_1=0 \\
    & &&& (9,0) &&&& \lambda - 12s_1=0 \\
    & &&& (12,0) &&&& \lambda - 20s_1=0 \\
    \hline
    & 2 &&& (1,1) &&&& \lambda^2 - 2s_1\lambda+3s_2 = 0 \\
    & &&& (4,1) &&&& \lambda^2 - 8s_1\lambda+12\left(s_1^2+ s_2\right) = 0\\
    & &&& (7,1) &&&& \lambda^2 - 18s_1\lambda+\left(72s_1^2+ 27s_2\right) = 0\\
    & &&& (10,1) &&&& \lambda^2 - 32s_1\lambda+\left(240s_1^2+ 48s_2\right) = 0\\
    \hline
    &3 &&& (2,2) &&&& \lambda^3 - 8s_1\lambda^2+\left(12s_1^2+ 28s_2\right)\lambda - \left(48s_1 s_2+80s_3\right) = 0 \\
    & &&& (5,2) &&&& \lambda^3 - 20 s_1\lambda^2+\left(108s_1^2+ 76s_2\right)\lambda -\left(144s_1^3 + 432s_1 s_2+224s_3\right) = 0 \\
    & &&& (8,2) &&&& \lambda^3 - 38 s_1\lambda^2+\left(432s_1^2+ 148s_2\right)\lambda - \left(1440s_1^3 + 1728s_1 s_2+440s_3\right) = 0\\
    \hline
    &4 &&& (3,3) &&&& \lambda^4 - 20 s_1\lambda^3+\left(108s_1^2+ 126s_2\right)\lambda^2 - \left(144s_1^3 + 972s_1 s_2+864s_3\right)\lambda +\\
    & &&& &&&& +\left(1080s_1^2s_2+945s_2^2+4320 s_1 s_3\right)= 0\\
    & &&& (6,3) &&&& \lambda^4 - 40 s_1\lambda^3+\left(508s_1^2+ 276s_2\right)\lambda^2 - \left(2304s_1^3 + 4872s_1 s_2+1944s_3\right)\lambda +\\
    & &&& &&&& +\left(2880s_1^4+17280s_1^2s_2+5760s_2^2+19440 s_1 s_3\right)= 0\\
    \hline
    &5 &&& (4,4) &&&& \lambda^5 - 40s_1 \lambda^4+(508s_1^2+396s_2)\lambda^3-(2304 s_1^3+7920 s_1 s_2 + 4752 s_3)\lambda^2+\\
    & &&& &&&& +(2880 s_1^4+ 38016 s_1^2 s_2 + 19008 s_2^2 + 76032 s_1s_3)\lambda-\\
    & &&& &&&& - (34560 s_1^3 s_2 + 241920 s_1^2 s_3 + 186624 s_2 s_3 +89856 s_1 s_2^2) = 0\\
    \hline
\end{align*}
Here `Mult.' stands for multiplicity and `Rep.' for~representation.

\stepcounter{appcounter}
\section{Polynomial solutions of Heun's equation}\label{polSolHeunEq}

In this Appendix we will derive conditions under which the Heun equation (\ref{HeunFinalForm}) has polynomial solutions. Let us search for solutions in the form
\begin{align}
    \mathcal{P}(x):=\sum_{j=0}^{\infty}c_j\, x^j\,,
\end{align}
i.e. requiring that $\mathcal{P}(x)$ be analytic at $x=0$. By direct substitution one shows that the coefficients $c_j$ satisfy the recurrence relation \cite[Chapter 31]{NIST}
\begin{align}
    &R_j\,c_{j+1}-(Q_j+q)\,c_j + P_j\, c_{j-1} = 0,\;\;\text{for}\;\;j\in \mathbb{N}\label{recurrenceJ}\\
    & -q\,c_0 = ap\,c_1\,.\label{recurrence0}
\end{align}
To obtain a polynomial solution of degree $\ell$ one terminates the series at $c_{\ell}$ by requiring that $c_{j} = 0$ for $j > \ell$. From (\ref{recurrenceJ}) evaluated at $j = \ell+1$ it follows that $P_{\ell+1} = 0$. Thus, $\alpha = -\ell$ or $\beta = -\ell$. Since the equations are symmetric under the permutation of $\alpha$ and~$\beta$, we may choose $\alpha = -\ell$. For the vector $\vec{c}^{\;\,\mathrm{T}}:=\begin{pmatrix}
    c_0 & c_1 & \dots & c_{\ell}
\end{pmatrix}$ one therefore has the equation 
\begin{align}
    \mathcal{M}_{\ell}\,\vec{c} = q\,\vec{c}\,,
\end{align}
which means that $q$ is an eigenvalue of the matrix $\mathcal{M}_{\ell}$.





\vspace{1cm}    
    \setstretch{0.8}
    \setlength\bibitemsep{5pt}
    \printbibliography

\end{document}